\documentclass[11pt,a4paper,reqno]{amsart}
\usepackage{amsmath}
\usepackage{amsfonts}
\usepackage{amssymb}
\usepackage{amsthm}
\usepackage{newlfont}
\usepackage{a4}
\usepackage{enumerate}
\usepackage[pdftex]{graphicx,color}

\theoremstyle{definition}
\newtheorem{defn}{Definition}[section]

\newtheorem{tvr}[defn]{Proposition}

\theoremstyle{remark}

\usepackage[centertags]{amsmath}
\usepackage{graphicx}
\usepackage{amsfonts}
\usepackage{amssymb}
\usepackage{amsthm}
\usepackage{bbm}
\usepackage{newlfont}
\usepackage{a4}
\usepackage{enumerate}

\newlength{\defbaselineskip}
\newcommand{\setlinespacing}[1]%
           {\setlength{\baselineskip}{#1 \defbaselineskip}}

\newcommand{\map}{\rightarrow}
\newcommand{\q}{\quad}
\renewcommand{\epsilon}{\varepsilon}
\renewcommand{\i}{\mathrm{i}}

\newcommand{\la}{\lambda}

\renewcommand{\rho}{\varrho}
\renewcommand{\phi}{\varphi}

\newcommand{\R}{{\mathbb{R}}}

\newcommand{\Com}{{\mathbb C}}

\newcommand{\Z}{\mathbb{Z}}

\newcommand{\set}[2]{\left\{#1  \mid #2 \right\}}

\newcommand{\abs}[1]{\left\vert#1\right\vert}
\newcommand{\wt}{\widetilde}

\newcommand{\setcomb}[2]{
\left\{
\begin{smallmatrix}
#1 \\ #2
\end{smallmatrix}
 \right\}  }

\begin{document}

\title[Two-dimensional (anti)symmetric exponential functions]
{Two-dimensional symmetric and antisymmetric generalizations
 of exponential and cosine functions}

\author{Ji\v{r}\'{i} Hrivn\'{a}k$^{1,2}$}
\author{Ji\v{r}\'{i} Patera$^1$}

\date{\today}
\begin{abstract}\

Properties of the four families of recently introduced special functions
of two real variables, denoted here by $E^\pm$, and $\cos^\pm$, are
studied. The superscripts $^+$ and $^-$ refer to the symmetric and
antisymmetric functions respectively. The functions are considered in all
details required for their exploitation in Fourier expansions of digital
data, sampled on square grids of any density and for general position of
the grid in the real plane relative to the lattice defined by the
underlying group theory. Quality of continuous interpolation, resulting
from the discrete expansions, is studied, exemplified and compared for
some model functions.
\end{abstract}\

\maketitle
\noindent
$^1$ Centre de recherches math\'ematiques,
         Universit\'e de Montr\'eal,
         C.~P.~6128 -- Centre ville,
         Montr\'eal, H3C\,3J7, Qu\'ebec, Canada; patera@crm.umontreal.ca\\
$^2$ Department of physics,
Faculty of nuclear sciences and physical engineering, Czech
Technical University, B\v{r}ehov\'a~7, 115 19 Prague 1, Czech
republic; jiri.hrivnak@fjfi.cvut.cz

\section{Introduction}

A large body of empirical evidence as well as theoretical experience has been accumulated in treating two-dimensional digital data owing to the large amount of data requiring treatment in practical contexts. We  approach the problem from the opposite direction. Our departure point, rather than a set of specific $2D$ data, are properties of new special functions in $n$ dimensions \cite{KPexp,KPtrig}. This presents certain advantages and disadvantages, some of which are listed in the concluding remarks. The most important distinction is undoubtedly the possibility to carry out any analysis of lattice data in the Fourier space rather than in the data space.

In this paper, common cosine transforms in one and two dimensions \cite{Strang,mart} are based on the symmetric group $S_2$ and $S_2\times S_2$ respectively. Our approach to $2D$ problems is built on group $S_3$ and on its affine extension. The special functions differ essentially from the product of trigonometric functions each depending on a single variable extending in mutually orthogonal directions in the case of $S_2\times S_2$. We start with the properties of the special functions of $S_n$  \cite{KPexp,KPtrig}.

Our aim here is to take advantage of the fact that, when working in $2D$, many specific properties of functions can be expressed explicitly and more detailed questions can be answered than in the case of the general dimension \cite{KPexp,KPtrig}. More precisely, we consider two families of symmetric and antisymmetric special functions of two real variables. They can be viewed as generalizations of the common exponential and cosine functions of one variable.

For each family, we study:
\newline
(i) The orthogonality of the functions on a lattice of chosen density within an appropriate finite region of the real Euclidean space $\R^2$, and the corresponding finite Fourier expansions of digital data functions sampled on such lattices.
\newline
(ii) The interpolation technique of the digital data used for the expansions into a continuous differentiable function. Interpolation quality is assessed as a function of the density of the lattice. Cosine transforms of types I, II, III, and IV \cite{Strang} can be generalized by situating the points of the sampling grid in a general position within the orthogonality region. It is important to point out those versions of the formalism by which the undesirable Gibbs effect can be avoided \cite{Ustina}.

The definition of the symmetric and antisymmetric exponential functions $E^+$ and $E^-$ is analogous to the definition of so called $C-$ and $S-$functions, where instead of the symmetric group $S_n$, the Weyl group $W$ of some compact semisimple Lie groups is used~\cite{KP1,KP2}. So called $E-$functions~\cite{KP3} which are based on even subgroup of the Weyl group $W$ are analogous to 'alternating' exponential functions~\cite{KPalte,KPalt}. In some cases, a direct relation between these two concepts can be found~\cite{NPT}. The coincidences of the orbit functions are due to the known isomorphism
of the symmetric group $S_n$ and the Weyl group of the Lie group $SU(n)$,
and of the alternating subgroup of $S_n$ and the even subgroup of the Weyl
group.

In Section~2, the antisymmetric and symmetric 2-dimensional functions are defined respectively as $2\times2$ determinants and permanents \cite{minc} of exponential functions of one variable. Properties of the antisymmetric functions (determinants) are considered first, namely their continuous and discrete orthogonality, the discrete Fourier transforms, and others. An exposition of analogous properties of the symmetric functions (permanents) follows. Examples of real and imaginary parts of the functions are shown in Fig.~1 and 2.

Section~3 is devoted to $2D$ interpolation by symmetric and antisymmetric exponential functions. General $2D$ interpolations are recalled, followed by a detailed study of interpolation by antisymmetric functions. Examples of the interpolation of the model function in Fig.~3 are shown in Fig.~4 for several densities of the sampling grid. Interpolation properties of the symmetric exponential function follow. Examples of the interpolation of the same model function are shown in Fig.~5. Interpolation errors are summarized in Table~1.

Section~4 parallels Sec,~2 and 3, except that the exponential functions are replaced by $2\times2$ antisymmetric (determinants) and symmetric (permanents) of cosine functions. Examples of antisymmetric and symmetric cosine functions are shown respectively in Fig.~6 and 8; interpolation examples can be found in Fig.~7 and 9. Interpolation quality in the above examples is compared in Table~1. The difference between the model function and its interpolations is integrated over the whole orthogonality region.

The last example deals with known undesirable effect of Fourier interpolation, referred to as the Gibbs effect \cite{Ustina}. In particular, it is shown that the Gibbs effect is absent in two versions of our $2D$ cosine transforms, as illustrated in Fig.~10.

The last section contains various related comments and remarks.

\medskip

\section{Two--dimensional (anti)symmetric exponential functions}
\subsection{Two--dimensional antisymmetric exponential functions}
\subsubsection{Definitions, symmetries and general properties}\

Two-dimensional antisymmetric exponential functions $E^-_{(\lambda,\mu)}:\R^2\map \Com$ have for $\lambda,\mu\in \R$ the following explicit form
\begin{align}
 E^-_{(\lambda,\mu)}(x,y)
     &=\left|\begin{smallmatrix}
     e^{2\pi i\lambda x}&e^{2\pi i\lambda y}\\
     e^{2\pi i\mu x}&e^{2\pi i\mu y}\\
     \end{smallmatrix}\right|
     =e^{2\pi i(\lambda x+\mu  y)}
     -e^{2\pi i(\lambda  y+\mu x)}
\end{align}

We observe that $E^-_{(\lambda,\lambda)}(x,y)=0$ and $E^-_{(\lambda,\mu)}(x,x)=0$. From the explicit formula we immediately obtain antisymmetry of $E^-_{(\lambda,\mu)}(x,y)$ with respect to the permutation of variables~$(x,y)$
\begin{equation}\label{expant}
 E^-_{(\lambda,\mu)}(y,x)=-E^-_{(\lambda,\mu)}(x,y).
\end{equation}
and, moreover, with respect to permutation of $(\la,\mu)$
\begin{equation}\label{expant2}
E^-_{(\lambda,\mu)}(x,y)= -E^-_{(\mu,\lambda)}(x,y).
\end{equation}
Therefore, we consider only such $E^-_{(\lambda,\mu)}$ with so called {\bf strictly dominant} $(\lambda,\mu)$, that is, pairs $(\lambda,\mu)$ with $\la>\mu$.

The functions $E^-_{(k,l)}$ with $k,l\in\Z$ have additional symmetries related to the periodicity of exponential function. One can verify directly that
\begin{equation}\label{expper}
    E^-_{(k,l)}(x+r,y+s)= E^-_{(k,l)}(x,y),\q r,s\in \Z
     \end{equation}

The relations (\ref{expant}) and (\ref{expper}) imply that it is sufficient to consider the functions $E^-_{(k,l)},\,k,l\in\Z$ on the so called {\bf fundamental domain} $F(S_2^{\mathrm{aff}})$~\cite{KPexp}. The following open 'triangle' can be chosen as the fundamental domain in dimension two,
\begin{equation}\label{fund}
F(S_2^{\mathrm{aff}})= \set{(x,y)\in(0,1)\times(0,1)}{
x>y}.
\end{equation}

We also have the following property for $a\in \R$
\begin{equation}\label{aexpshift}
 E^-_{(k,l)}(x+a,y+a)= e^{2\pi i(k+l)a}E^-_{(k,l)}(x,y).
\end{equation}

\subsubsection{Continuous orthogonality}\

The functions $E^-$ are mutually orthogonal on the fundamental domain, i.e.
\begin{equation*}
    \int_{F(S_2^{\mathrm{aff}})} E^-_{(k,l)}(x,y)
\overline{E^-_{(k',l')}(x,y)}\,dx dy=
\delta_{kk'}\delta_{ll'}, \q k,l,k',l' \in\Z,\, k>l,k'>l'.  \end{equation*}
Therefore, every function $f:\R^2\map \Com$ that is antisymmetric $f(x,y)=-f(y,x)$ and periodic $f(x+r,y+s)= f(x,y),\, r,s\in \Z$ and has continuous derivatives can be expanded in the antisymmetric exponential functions $E^-_{(k,l)}$:
\begin{equation}
f(x,y)=\sum_{\setcomb{k,l\in \Z}{k>l}} {\wt c}_{kl} E^-_{(k,l)}(x,y)
,\q
{\wt c}_{kl} = \int_{F(S^{\mathrm{aff}}_2)} f(x,y)
\overline{E^-_{(k,l)}(x,y)}\,dx\, dy.
\end{equation}

\subsubsection{Solutions of the Laplace equation}\

The functions $E^-$ are solutions of the Laplace equation
\begin{equation}\label{laplacem}
\left(\frac{\partial^2}{\partial x^2}+\frac{\partial^2}{\partial y^2}\right)E^-_{(k,l)}(x,y)=-4\pi^2(k^2+l^2)E^-_{(k,l)}(x,y)
\end{equation}
and, moreover, of the equation
\begin{equation*}
\frac{\partial^2}{\partial x^2}\frac{\partial^2}{\partial y^2}E^-_{(k,l)}(x,y)=16\pi^4k^2l^2E^-_{(k,l)}(x,y),
\end{equation*}
which is algebraically independent. The functions $E^-$ satisfy the condition $E^-_{(k,l)}(x,y)=0$ on the boundary $x=y$.

\subsubsection{Discrete orthogonality}\

Discrete orthogonality of antisymmetric exponential functions over the grid of form $(x_m,y_n)=(m/N,n/N),\, m,n\in \{0,\dots,N-1\}$, $m>n$ was proved in~\cite{KPexp}. The positive integer $N$ sets the density of the grid inside $F(S_2^{\mathrm{aff}})$. For applications, it may be convenient to consider orthogonality over more a general type of grid. In addition to the parameter $N$, we parameterize the grid by parameters $a\in \R$ and $b\in [0,1]$. The equidistant $N(N-1)/2$--point grid $L^-_{a,b,N,1}$ is given by
$$L^-_{a,b,N,1}=\left\{(x_m,y_n)\,|\,m>n,\,m,n=0\dots N-1  \right \}$$
where
\begin{equation}\label{xmyn}
    (x_m,y_n)=\left( a+\frac{m+b}{N},a+\frac{n+b}{N}  \right).
\end{equation}
Using the property (\ref{aexpshift}), we observe that the orthogonality relations from~\cite{KPexp} also hold over the grid $L^-_{a,b,N,1}$:
\begin{equation}\label{adortho}
\sum_{\setcomb{m,n=0}{m>n}}^{N-1} E^-_{(k,l)}(x_m,y_n)\overline{E^-_{(k',l')}(x_m,y_n)}=N^2\delta_{kk'}\delta_{ll'},
\end{equation}
where $k,l,k',l'\in \{0,\dots,N-1\}, k>l,k'>l'$. The set of functions $E^-_{(k,l)}$ with $k,l\in \{0,\dots,3\}, k>l$ of mutually orthogonal functions on the grid $L^-_{a,b,4,1}$ is depicted in Figure \ref{FEm}.

Note that for $N$ odd, $N=2M+1$, the discrete orthogonality relations (\ref{adortho}) also hold for the set of functions $E^-_{(k,l)}$ with $k,l\in \{-M,\dots,M\}, k>l$.
 \begin{figure}[!ht]
\resizebox{2.4cm}{!}{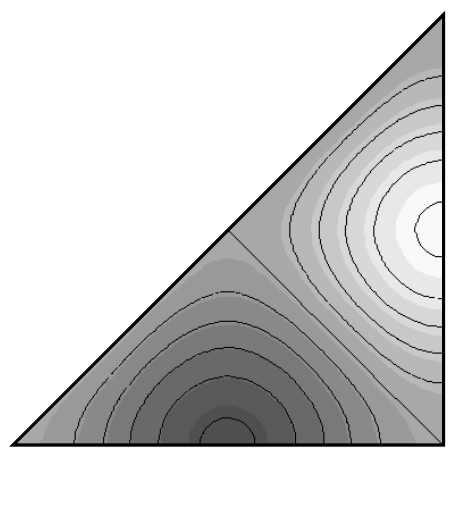}\hspace{22pt}
\resizebox{2.4cm}{!}{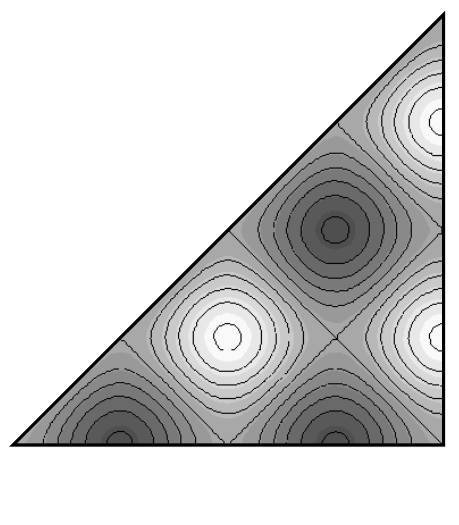}\hspace{22pt}
\resizebox{2.4cm}{!}{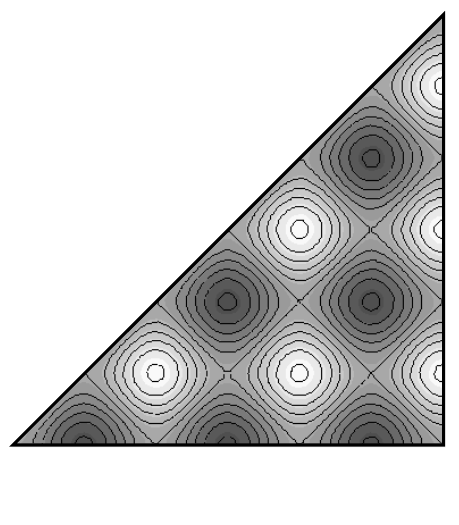}
\resizebox{2.4cm}{!}{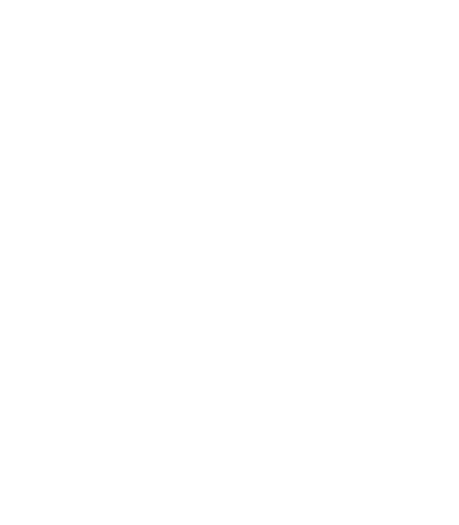}
\\\vspace{2pt}
\resizebox{2.4cm}{!}{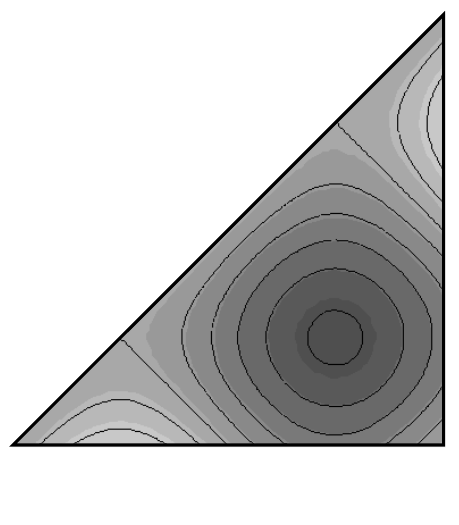}\hspace{22pt}
\resizebox{2.4cm}{!}{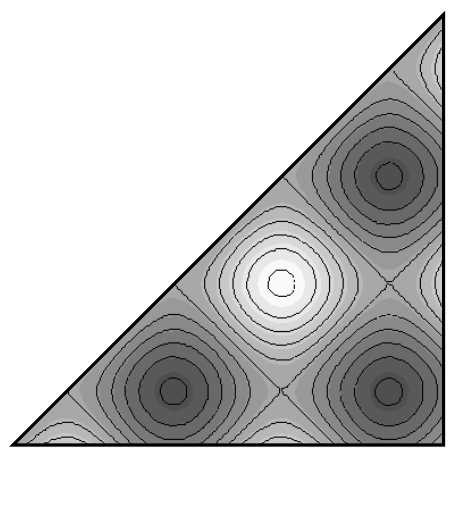}\hspace{22pt}
\resizebox{2.4cm}{!}{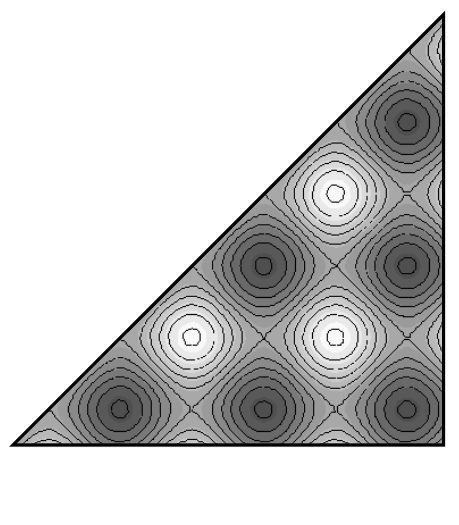}
\resizebox{2.4cm}{!}{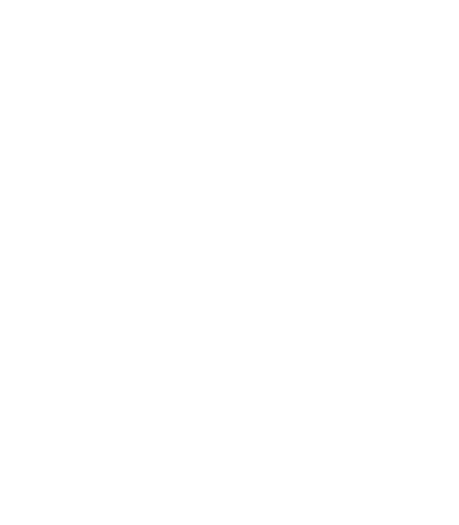}\\
\vspace{18pt}
\resizebox{2.4cm}{!}{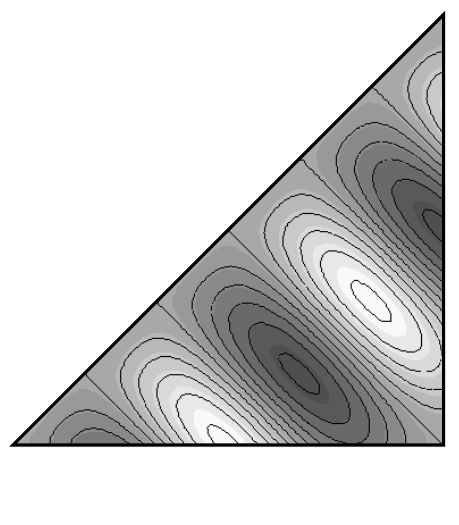}\hspace{22pt}
\resizebox{2.4cm}{!}{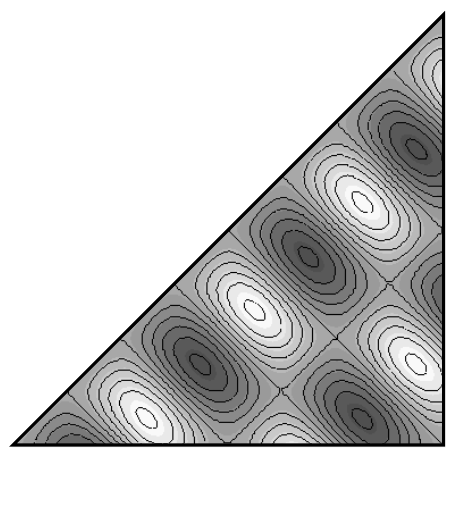}\hspace{22pt}
\resizebox{2.4cm}{!}{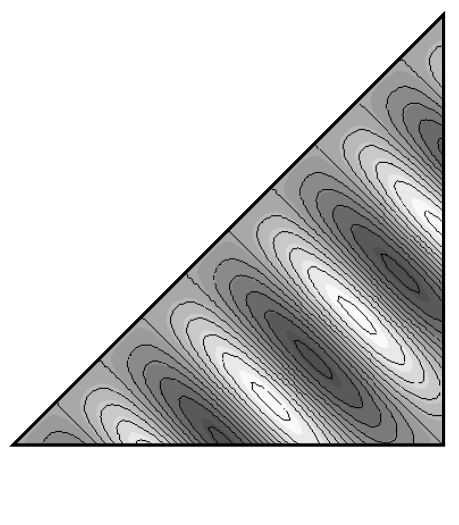}
\resizebox{2.4cm}{!}{\input{real.pdf_t}}
\\\vspace{2pt}
\resizebox{2.4cm}{!}{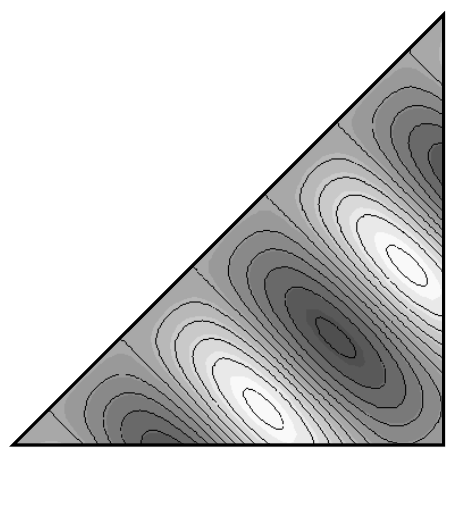}\hspace{22pt}
\resizebox{2.4cm}{!}{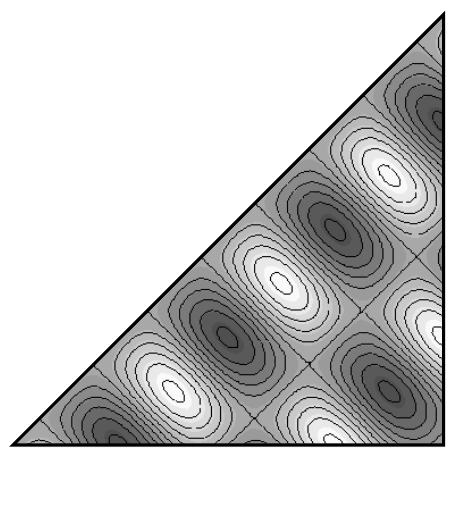}\hspace{22pt}
\resizebox{2.4cm}{!}{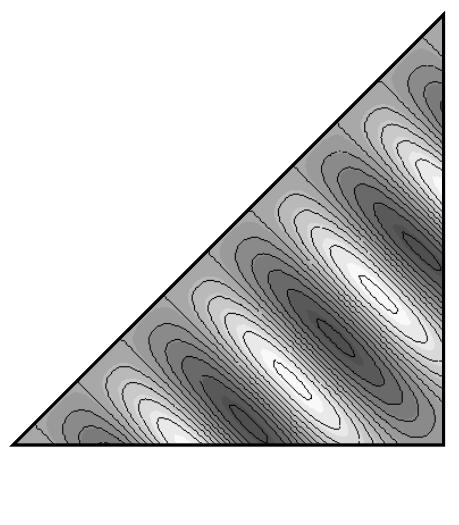}
\resizebox{2.4cm}{!}{\input{imaginary.pdf_t}}
\caption{The contour plots of continuous functions $E^-_{(k,l)}$, $k,l\in \{0,\dots,3\}, k>l$, which are discretely pairwise orthogonal on the grid $L^-_{a,b,4,1}$.}\label{FEm}
\end{figure}

\subsubsection{Antisymmetric discrete Fourier transform}\

Suppose we have a discrete function $f:L^-_{a,b,N,1}\map \Com $ defined on the grid $L^-_{a,b,N,1}$. The {\it antisymmetric discrete Fourier transform} of $f$ over $L^-_{a,b,N,1}$ is given by
\begin{equation}\label{abetas}
\beta^-_{kl}=\frac{1}{N^2}\sum_{\setcomb{m,n=0}{m>n}}^{N-1}f(x_m,y_n)\overline{E^-_{(k,l)}(x_m,y_n)} \q k>l,\,k,l=0\dots N-1.
\end{equation}
Orthogonality relation (\ref{adortho}) immediately gives the inverse transform of $N(N-1)/2$ coefficients $\beta^-_{kl}$:
\begin{equation}\label{aphase}
    f(x_m,y_n)=\sum_{\setcomb{k,l=0}{k>l}}^{N-1}\beta^-_{kl}E^-_{(k,l)}(x_m,y_n).
\end{equation}

\subsection{Two--dimensional symmetric exponential functions}\

Two-dimensional symmetric exponential functions $E^+_{(\lambda,\mu)}:\R^2\map \Com$ have for $\lambda,\mu\in \R$  the following explicit form
  \begin{align}
 E^+_{(\lambda,\mu)}(x, y)
     &=\left|\begin{smallmatrix}
     e^{2\pi i\lambda x}&e^{2\pi i\lambda y}\\
     e^{2\pi i\mu x}&e^{2\pi i\mu y}\\
     \end{smallmatrix}\right|^+
     =e^{2\pi i(\lambda x+\mu y)}
     +e^{2\pi i(\lambda y+\mu x)}
\end{align}
From the explicit formula we immediately obtain the symmetry of $E^+_{(\lambda,\mu)}(x,y)$ with respect to permutation of variables~$(x,y)$ \begin{equation}\label{expsym}
 E^+_{(\lambda,\mu)}(y,x)=E^+_{(\lambda,\mu)}(x,y).
\end{equation}
and, moreover, with respect to the permutation of $(\la,\mu)$
\begin{equation}\label{expsym2}
E^+_{(\lambda,\mu)}(x,y)= E^+_{(\mu,\lambda)}(x,y).
\end{equation}
Therefore, we consider only such $E^+_{(\lambda,\mu)}$ with so called {\bf dominant} $(\lambda,\mu)$, that is, pairs $(\lambda,\mu)$ with $\la\geq\mu$. Functions $E^+_{(k,l)}$ with $k,l\in\Z$ have additional symmetries related to the periodicity of exponential function. One can verify directly that
\begin{equation}\label{exppers}
    E^+_{(k,l)}(x+r,y+s)= E^+_{(k,l)}(x,y),\q r,s\in \Z
     \end{equation}
The relations (\ref{expsym}) and (\ref{exppers}) imply that it is sufficient to consider the functions $E^+_{(k,l)},\,k,l\in\Z$ on the closure of the fundamental domain $F(S_2^{\mathrm{aff}})$~\cite{KPexp}.

We also have the following property for $a\in \R$
\begin{equation}\label{sexpshift}
 E^+_{(k,l)}(x+a,y+a)= e^{2\pi i(k+l)a}E^+_{(k,l)}(x,y).
\end{equation}

\subsubsection{Continuous orthogonality}\

The functions $E^+$ are mutually orthogonal on $\overline{F(S_2^{\mathrm{aff}})}$, i.e.
\begin{equation*}
    \int_{F(S_2^{\mathrm{aff}})} E^+_{(k,l)}(x,y)
\overline{E^+_{(k',l')}(x,y)}\,dx\,dy=
G_{kl}\delta_{kk'}\delta_{ll'},\q k,l,k',l' \in\Z,\, k\geq l,k'\geq l' \end{equation*}
where symbol $G_{kl}$ is defined by
\begin{equation}
G_{kl}=\begin{cases} 2 & \text{if $k=l$} \\ 1 & \text{otherwise}. \end{cases}
\end{equation}

Every function $f:\R^2\map \Com$ that is symmetric $f(x,y)=f(y,x)$ and periodic $f(x+r,y+s)= f(x,y),\, r,s\in \Z$ and has continuous derivatives can be expanded in the antisymmetric exponential functions $E^+_{(k,l)}$:
\begin{equation}
f(x,y)=\sum_{\setcomb{k,l\in \Z}{k\geq l}} {\wt c}_{kl} E^+_{(k,l)}(x,y)
,\q
{\wt c}_{kl} = G_{kl}^{-1}\int_{F(S^{\mathrm{aff}}_2)} f(x,y)
\overline{E^+_{(k,l)}(x,y)}\,dx\, dy.
 \end{equation}

\subsubsection{Solutions of the Laplace equation}\

The functions $E^+$ are solutions of the Laplace equation
\begin{equation}\label{laplacep}
\left(\frac{\partial^2}{\partial x^2}+\frac{\partial^2}{\partial y^2}\right)E^+_{(k,l)}(x,y)=-4\pi^2(k^2+l^2)E^+_{(k,l)}(x,y)
\end{equation}
and moreover of the equation
\begin{equation*}
\frac{\partial^2}{\partial x^2}\frac{\partial^2}{\partial y^2}E^+_{(k,l)}(x,y)=16\pi^4k^2l^2E^+_{(k,l)}(x,y).
\end{equation*}
The functions $E^+$ satisfy the condition $$\frac{\partial}{\partial \mathbf{n} }E^+_{(k,l)}(x,y)=0,$$ where $\mathbf{n}$ is the normal to the boundary $x=y$.

\subsubsection{Discrete orthogonality}\

Discrete orthogonality of symmetric exponential functions over the grid of form $(x_m,y_n)=(m/N,n/N),$ $ m,n\in \{0,\dots,N-1\}$, $m\geq n$ was proved in~\cite{KPexp}. We consider the orthogonality over a more general type of grid. Besides the parameter $N$, we parameterize the grid by parameters $a\in \R$ and $b\in [0,1]$. The
equidistant $N(N+1)/2$--point grid $L^+_{a,b,N,1}$ is given by
$$L^+_{a,b,N,1}=\left\{(x_m,y_n)\,|\,m\geq n,\,m,n=0\dots N-1  \right \}$$
where $(x_m,y_n)$ are given by (\ref{xmyn}).
Using the property (\ref{sexpshift}), we observe that the orthogonality relations from~\cite{KPexp} also hold over the grid $L^+_{a,b,N,1}$:
\begin{equation}\label{sdortho}
\sum_{\setcomb{m,n=0}{m\geq n}}^{N-1} G_{mn}^{-1}E^+_{(k,l)}(x_m,y_n)\overline{E^+_{(k',l')}(x_m,y_n)}=G_{kl} N^2\delta_{kk'}\delta_{ll'},
\end{equation}
where $k,l,k',l'\in \{0,\dots,N-1\}, k\geq l,k'\geq l'$. The set of functions $E^+_{(k,l)}$ with $k,l\in \{0,1,2\}, k\geq l$ of mutually orthogonal functions on the grid $L^+_{a,b,3,1}$ is depicted in Figure \ref{FEp}.

Note that for $N$ odd, $N=2M+1$, the discrete orthogonality relations (\ref{sdortho}) also hold for the set of functions $E^+_{(k,l)}$ with $k,l\in \{-M,\dots,M\}, k\geq l$.

\begin{figure}[!ht]
\resizebox{2.4cm}{!}{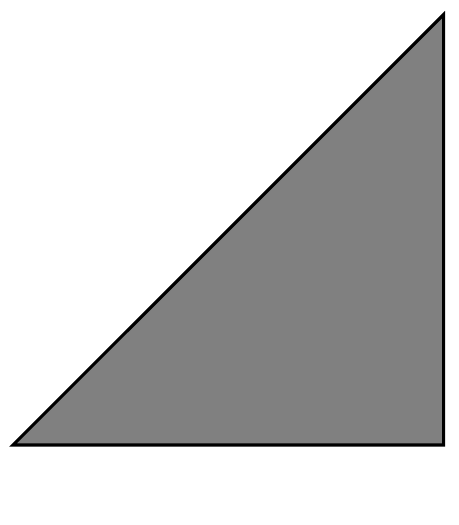}\hspace{22pt}
\resizebox{2.4cm}{!}{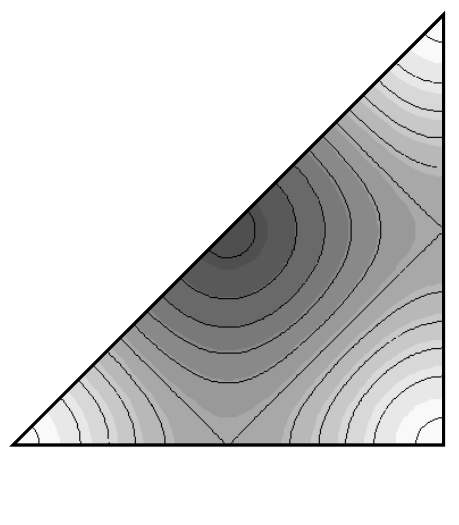}\hspace{22pt}
\resizebox{2.4cm}{!}{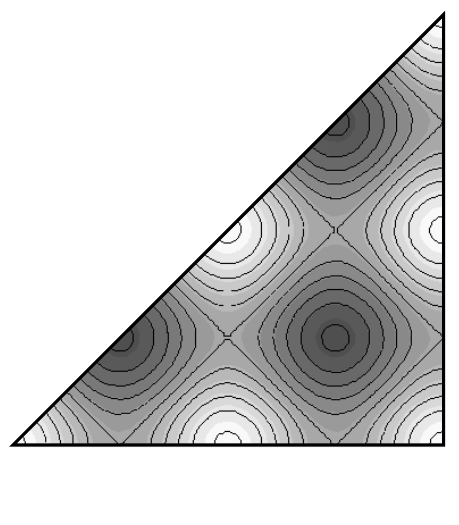}
\resizebox{2.4cm}{!}{\input{real.pdf_t}}
\\\vspace{2pt}
\resizebox{2.4cm}{!}{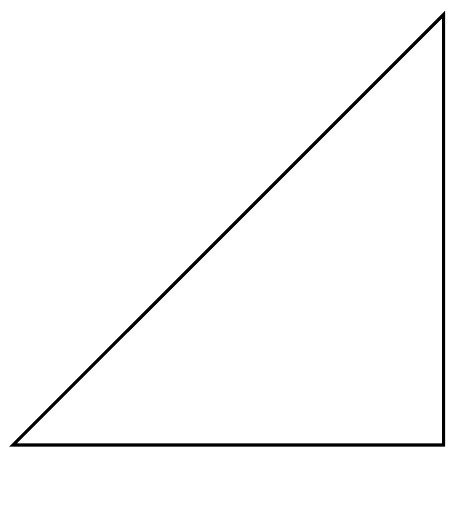}\hspace{22pt}
\resizebox{2.4cm}{!}{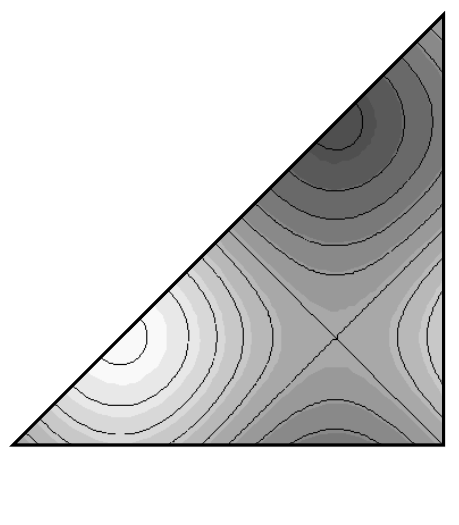}\hspace{22pt}
\resizebox{2.4cm}{!}{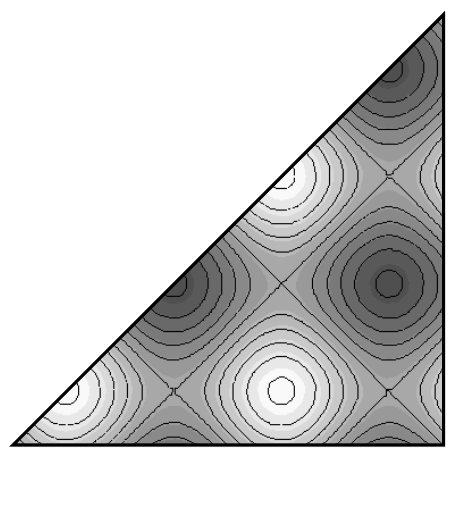}
\resizebox{2.4cm}{!}{\input{imaginary.pdf_t}}\\
\vspace{18pt}
\resizebox{2.4cm}{!}{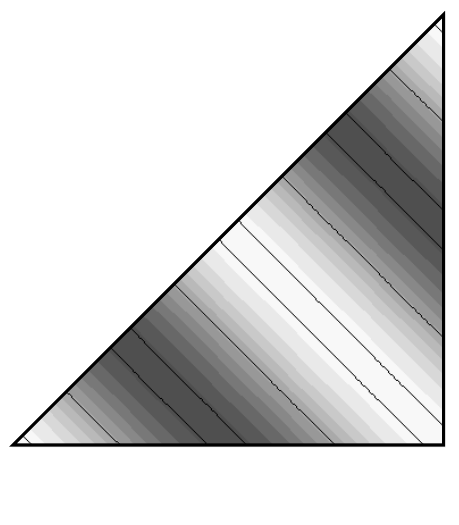}\hspace{22pt}
\resizebox{2.4cm}{!}{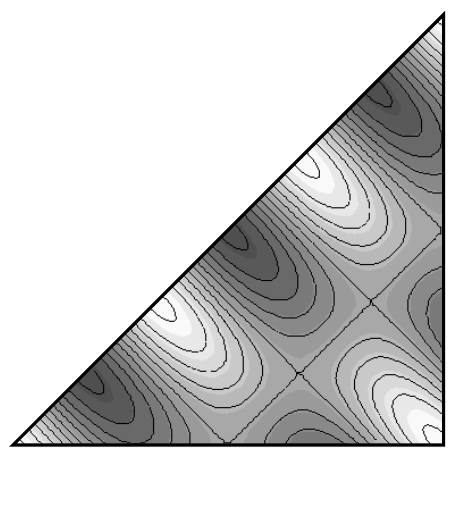}\hspace{22pt}
\resizebox{2.4cm}{!}{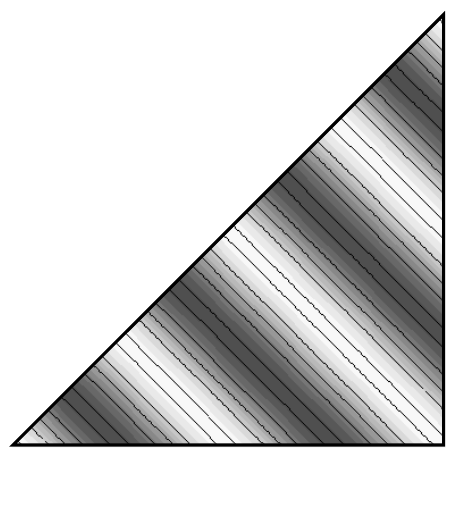}
\resizebox{2.4cm}{!}{\input{real.pdf_t}}
\\\vspace{2pt}
\resizebox{2.4cm}{!}{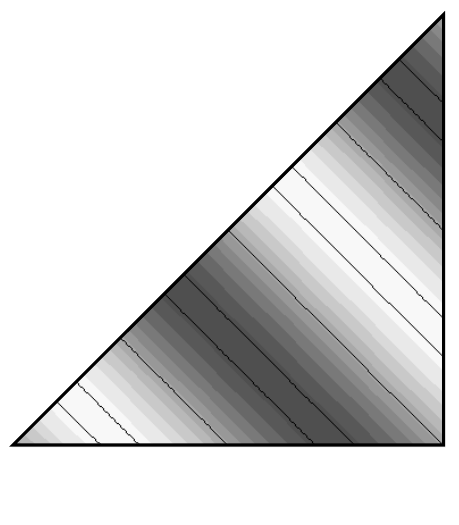}\hspace{22pt}
\resizebox{2.4cm}{!}{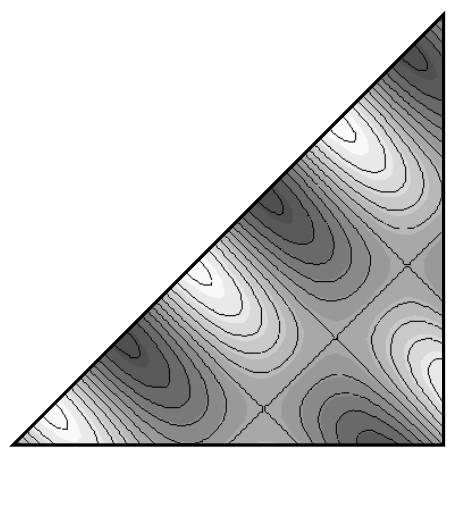}\hspace{22pt}
\resizebox{2.4cm}{!}{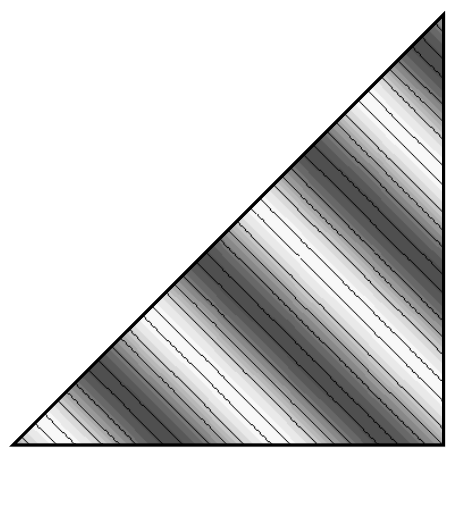}
\resizebox{2.4cm}{!}{\input{imaginary.pdf_t}}
\caption{The contour plots of the continuous functions $E^+_{(k,l)}$, $k,l\in \{0,1,2\}, k\geq l$, which are discretely pairwise orthogonal on the grid $L^+_{a,b,3,1}$.}\label{FEp}
\end{figure}
\subsubsection{Symmetric discrete Fourier transform}\

Suppose we have a discrete function $f:L^+_{a,b,N,1}\map \Com $ defined on the grid $L^+_{a,b,N,1}$. The {\it symmetric discrete Fourier transform} of $f$ over $L^+_{a,b,N,1}$ is given by
\begin{equation}\label{sbetas}
\beta^+_{kl}=\frac{1}{G_{kl}N^2}\sum_{\setcomb{m,n=0}{m\geq n}}^{N-1}G_{mn}^{-1} f(x_m,y_n)\overline{E^+_{(k,l)}(x_m,y_n)} \q k\geq l,\,k,l=0\dots N-1.
\end{equation}
The orthogonality relation (\ref{adortho}) immediately gives the inverse transform of $N(N+1)/2$ coefficients $\beta^+_{kl}$:
\begin{equation}\label{sphase}
    f(x_m,y_n)=\sum_{\setcomb{k,l=0}{k\geq l}}^{N-1}\beta^+_{kl}E^+_{(k,l)}(x_m,y_n)
\end{equation}

\section{Interpolation by (anti)symmetric exponential functions}

\subsection{General two--dimensional trigonometric interpolation}\

Let us consider a symmetrically placed square in $\R^2$ with a side of $T\in \R$. For $a\in \R$, the square is given by
\begin{equation}\label{square}
K_{[a,a']}=[a,a']\times [a,a']
\end{equation}
Let us also select an arbitrary natural number $N$ and parameter $b\in [0,1]$ and consider a symmetrically placed $N^2$-point grid $L_{a,b,N,T}=\left\{(x_m,y_n)\,|\,m,n=0\dots N-1  \right \}\subset K_{[a,a+T]} $
where
\begin{equation*}
    (x_m,y_n)=\left( a+\frac{m+b}{N}T,a+\frac{n+b}{N}T  \right).
\end{equation*}
Suppose we have a given function $f:K_{[a,a+T]}\map \Com$ and a set of points $L_{a,b,N,T}\subset K_{[a,a+T]}$. In the following, we distinguish two cases, namely $N=2M+1$, or $N=2M$. The {\it (trigonometric) interpolation problem} can be formulated in the following way: find a {\bf trigonometric interpolating polynomial} of the form \begin{equation}\label{trigT}
\psi_{N,T}(x,y)=\sum_{k,l
  =-M}^{M}c_{kl}e^{2\pi\i k \frac{x}{T}}e^{2\pi\i l\frac{y}{T}},\qquad x,y\in\mathbb R
\end{equation}
such that it coincides with $f$ on the grid $L_{a,b,N,T}$, which means satisfying for all $(x_m,y_n)\in L_{a,b,N,T}$ the condition $\psi_{N,T}(x_m,y_n)=f(x_m,y_n)$. Henceforward we set for simplicity $T=1$, i.e. we have the trigonometric interpolating polynomial $\psi_{N}\equiv\psi_{N,1}$ of the form
\begin{equation}\label{trig}
    \psi_N(x,y)=\sum_{k,l=-M}^{M}c_{kl}e^{2\pi\i k x}e^{2\pi\i ly},\qquad x,y\in\mathbb R
\end{equation}
satisfying on $K_{[a,a+1]}$
\begin{equation}\label{trig2}
    \psi_{N}(x_m,y_n)=f(x_m,y_n), \q m,n=0\dots N-1.
\end{equation}
Note that, in all of the following formulas, we can always recover an arbitrary size $T$ simply by linear transformation $$(x,y)\map \left(\frac{x}{T},\frac{y}{T}\right). $$
For $N=2M+1$, the trigonometric interpolating polynomial $\psi_{N}$ has $(2M+1)^2=N^2$ unknown coefficients $c_{kl}$, to which correspond $N^2$ constraints (\ref{trig2}). For $N=2M$, denoting $\tau_N=e^{2\pi \i (Na+b)}$, we assume further $4M+1$ conditions
\begin{equation}\label{assum}
\begin{split}
 c_{k,-M} &=\tau_N c_{k,M} ,\q k=-M\dots M-1 \\
 c_{-M,l}&= \tau_N c_{M,l},\q l=-M\dots M
\end{split}
\end{equation}
and we have $(2M+1)^2-(4M+1)=(2M)^2=N^2$ unknown coefficients $c_{kl}$ corresponding to $N^2$ constraints (\ref{trig2}).

In order to analyze the system of linear equations (\ref{trig2}),  we introduce an $N\times (2M+1)$ matrix $$V_M(x_0,\dots,x_{N-1}):=\begin{pmatrix}e^{2\pi\i (-M)x_0} & e^{2\pi\i (-M+1)x_0} & \cdots & 1 & \cdots & e^{2\pi\i (M-1)x_0} & e^{2\pi\i (M)x_0} \\
e^{2\pi\i (-M)x_1} & e^{2\pi\i (-M+1)x_1} & \cdots & 1 & \cdots & e^{2\pi\i (M-1)x_1} & e^{2\pi\i (M)x_1} \\
\vdots & \vdots & \ddots & \vdots & \vdots &  \vdots & \vdots \\
e^{2\pi\i (-M)x_{N-1}} & e^{2\pi\i (-M+1)x_{N-1}} & \cdots & 1 & \cdots & e^{2\pi\i (M-1)x_{N-1}} & e^{2\pi\i (M)x_{N-1}}
\end{pmatrix}.$$
and the product
$$W_N\equiv W_N(x_0,\dots,x_{N-1}) \equiv W_N(y_0,\dots,y_{N-1}):=\prod_{\setcomb{m,n=0}{m>n}}^{N-1}(e^{2\pi\i x_m}-e^{2\pi\i x_n} ).    $$

The coefficient matrix of the system (\ref{trig2}) can be written for $N=2M+1$ as
\begin{equation}\label{cmat}
 V_M(x_0,\dots,x_{N-1}) \otimes V_M(y_0,\dots,y_{N-1}).
\end{equation}
The matrix (\ref{cmat}) is Vandermonde-like and its determinant is
$$ \det [V_M(x_0,\dots,x_{N-1}) \otimes V_M(y_0,\dots,y_{N-1})]= \tau_N^{-N(N-1)}W_N^{2N}.$$
For $N=2M$, we add to the matrix (\ref{cmat}) an additional block of rows (\ref{assum}). The determinant of the resulting matrix is equal to
$$-4^{N}\tau_N^{-N^{2}+2N+1}W_N^{2N}. $$
Note that both determinants are always non-zero. Thus, the solution of the interpolation problem always exists and is unique.

The coefficients $c_{kl}$ are
given for $N=2M+1$ by
\begin{equation}\label{interodd}
c_{kl}=\frac{1}{N^2}\sum_{m,n=0}^{N-1}f(x_m,y_n)e^{-2\pi\i k x_m}e^{-2\pi\i ly_n}
\end{equation}
and, introducing the symbol $g_{k,M}$ by
\begin{equation}
g_{k,M}=\begin{cases}\frac12 & \text{if $k=-M,M$} \\ 1 & \text{otherwise}, \end{cases}
\end{equation}
one can write for $N=2M$
\begin{equation}\label{intereven}
 c_{kl}=\frac{g_{k,M}g_{l,M}}{N^2}\sum_{m,n=0}^{N-1}f(x_m,y_n)e^{-2\pi\i k x_m}e^{-2\pi\i ly_n}.
\end{equation}

\subsection{Antisymmetric interpolation}\

For interpolation with antisymmetric exponential functions, we consider the triangle $K_{[a,a+1]}^-$
inside the square $K_{[a,a']}$:
\begin{equation}\label{afunds}
K_{[a,a+1]}^-= \set{(x,y)\in[a,a+1]\times[a,a+1]}{
x> y}.
\end{equation}

For a given function $f:K_{[a,a+1]}^-\map \Com$ and a set of points $L^-_{a,b,N,1}\subset K_{[a,a+1]}^-$ we define an {\it antisymmetric trigonometric interpolating function}
\begin{equation}\label{atrig}
    \psi^-_{N}(x,y)=\sum_{\setcomb{k,l=-M}{k>l}}^{M}c^-_{kl}E^-_{(k,l)}(x,y),\q x,y\in \R
\end{equation}
satisfying
\begin{equation}\label{atrig2}
    \psi^-_{N}(x_m,y_n)=f(x_m,y_n), \q m>n,\,  m,n=0\dots N-1.
\end{equation}

For $N=2M+1$, the antisymmetric interpolating function $\psi^-_{N}$ has $N(N-1)/2$ unknown coefficients $c_{kl}$, to which correspond $N(N-1)/2$ constraints (\ref{atrig2}). For $N=2M$, we assume further $2M$ conditions
\begin{equation}\label{aassum}
\begin{split}
 c^-_{l,-M}&= -\tau_N c^-_{M,l},\q l=-M+1\dots M-1 \\
 c^-_{M,-M}&=0
\end{split}
\end{equation}
and we have $(2M+1)2M/2-2M=N(N-1)/2$ unknown coefficients $c^-_{kl}$ corresponding to $N(N-1)/2$ constraints (\ref{atrig2}).
\begin{tvr}\label{aip}
There exists a unique antisymmetric interpolating function (\ref{atrig}) satisfying (\ref{atrig2}). The coefficients $c^-_{kl}$ are
given for $N=2M+1$ by
\begin{equation}\label{ainterodd}
c^-_{kl}=\frac{1}{N^2}\sum_{\setcomb{m,n=0}{m>n}}^{N-1}f(x_m,y_n)\overline{E^-_{(k,l)}(x_m,y_n)}
\end{equation}
and for $N=2M$ assuming (\ref{aassum}) by
\begin{equation}\label{aintereven}
 c^-_{kl}=\frac{g_{k,M}g_{l,M}}{N^2}\sum_{\setcomb{m,n=0}{m>n}}^{N-1}f(x_m,y_n)\overline{E^-_{(k,l)}(x_m,y_n)}.
\end{equation}
\end{tvr}
\begin{proof}
For the function $f:K_{[a,a+1]}^-\map \Com$, we define its antisymmetric extension $Af:K_{[a,a+1]}\map \Com$ by the formula
\begin{equation}\label{aext}
 Af(x,y)=\begin{cases}
f(x,y) & x>y\\
0 & x=y \\
-f(y,x) & x<y.
\end{cases}
\end{equation}
We obtain a unique trigonometric interpolating polynomial $A\psi_{N}$ of the form
\begin{equation}\label{AAtrig}
    A\psi_N(x,y)=\sum_{k,l=-M}^{M}(Ac)_{kl}e^{2\pi\i k x}e^{2\pi\i ly},\q x,y\in \R
\end{equation}
satisfying
\begin{equation}\label{AAtrig2}
    A\psi_{N}(x_m,y_n)=Af(x_m,y_n), \q m,n=0\dots N-1.
\end{equation}
According to (\ref{interodd}), the coefficients $(Ac)_{kl}$ are
given for $N=2M+1$ by
\begin{equation}\label{AAcoeff}
 (Ac)_{kl}=\frac{1}{N^2}\sum_{m,n=0}^{N-1}Af(x_m,y_n)e^{-2\pi\i k x_m}e^{-2\pi\i ly_n}.
\end{equation}
Note that since $Af$ is antisymmetric we have $(Ac)_{kl}=-(Ac)_{lk}$, and so, $(Ac)_{kk}=0$. Taking this into account, the formula (\ref{AAtrig}) can be rewritten as
\begin{equation}\label{jjjh}
\begin{split}
A\psi_N(x,y)=&\sum_{\setcomb{k,l=-M}{k>l}}^{M}(Ac)_{kl}e^{2\pi\i k x}e^{2\pi\i ly}-\sum_{\setcomb{k,l=-M}{k<l}}^{N-1}(Ac)_{lk}e^{2\pi\i k x}e^{2\pi\i ly}\\
 =&\sum_{\setcomb{k,l=-M}{k>l}}^{M}(Ac)_{kl}E^-_{(k,l)}(x,y)
\end{split}
\end{equation}

Using the symmetricity of the interpolating grid $x_m=y_m$ and the antisymmetricity of $Af$ we rewrite (\ref{AAcoeff})
\begin{equation}\label{AAint}
\begin{split}
(Ac)_{kl}=&\frac{1}{N^2}\sum_{\setcomb{m,n=0}{m>n}}^{N-1}Af(x_m,y_n)e^{-2\pi\i k x_m}e^{-2\pi\i ly_n}-\frac{1}{N^2}\sum_{\setcomb{m,n=0}{m<n}}^{N-1}f(y_n,x_m)e^{-2\pi\i k x_m}e^{-2\pi\i ly_n}\\
& =\frac{1}{N^2}\sum_{\setcomb{m,n=0}{m>n}}^{N-1}f(x_m,y_n)\overline{E^-_{(k,l)}(x_m,y_n)}:=c^-_{kl}
\end{split}
\end{equation}
and analogously for $N=2M$.

The existence of two different antisymmetric interpolation polynomials would imply the existence of two different interpolating polynomials satisfying (\ref{AAtrig}), (\ref{AAtrig2}) --- a contradiction.
\end{proof}

\subsubsection{Calculation of the coefficients $c^-_{kl}$}\

Instead of the direct calculation of the coefficients $c^-_{kl}$, the antisymmetric discrete Fourier transform (\ref{abetas}) can be used, and the resulting coefficients $\beta^-_{kl}$ transformed to $c^-_{kl}$'s. By direct comparison of (\ref{ainterodd}), (\ref{aintereven}) to (\ref{abetas}), we obtain for $N=2M+1$
\begin{equation}
\begin{split}
c^-_{kl}=\beta^-_{kl},\q &k>l, k,l=0\dots M \\
c^-_{k,-l}=-\tau_N\beta^-_{N-l,k},\q & k=0\dots M,\,l=1\dots M\\
c^-_{-k,-l}=\tau_N^2\beta^-_{N-k,N-l}, \q &k<l,\, k,l=1\dots M
\end{split}
\end{equation}
and for $N=2M$
\begin{equation}\label{abetastoceven}
\begin{split}
c^-_{kl}=g_{k,M}g_{l,M}\beta^-_{kl},\q &k>l,\,k,l=0\dots M \\
c^-_{k,-l}=-g_{k,M}\tau_N\beta^-_{N-l,k},\q & k=0\dots M,\,l=1\dots M-1\\
c^-_{-k,-l}=\tau_N^2\beta^-_{N-k,N-l}, \q &k<l,\, k,l=1\dots M-1
\end{split}
\end{equation}
The formula (\ref{abetastoceven}) determines $(2M)(2M-1)/2$ coefficients; the rest of the $2M$ coefficients are determined via the relations (\ref{aassum}).

\subsubsection{Trigonometric form of $\psi^-_N(x,y)$}\

The antisymmetric interpolating polynomial (\ref{atrig}) can be brought to its 'trigonometric form'.
Introducing the symbol $h_{k}$ via the relation
\begin{equation}\label{hk}
h_{k}=\begin{cases}\frac12 & \text{if $k=0$} \\ 1 & \text{otherwise} \end{cases}
\end{equation}
we have
\begin{equation}\label{base}
\begin{split}
    \psi^-_N(x,y)=&\sum_{\setcomb{k,l=0}{k>l}}^{M}h_{k}h_{l}[A^-_{kl}(\cos 2\pi k x \cos 2\pi ly-\cos 2\pi l x \cos 2\pi ky )\\ +&B^-_{kl}(\sin 2\pi k x \cos 2\pi ly-\cos 2\pi l x \sin 2\pi ky)+\\
                +& C^-_{kl}(\cos 2\pi k x \sin 2\pi ly-\sin 2\pi l x \cos 2\pi ky)\\ +&D^-_{kl}(\sin 2\pi k x \sin 2\pi ly-\sin 2\pi l x \sin 2\pi ky)]
\end{split}
\end{equation}
where
\begin{equation}\label{aA}
\begin{split}
    A^-_{kl}=& c^-_{k,l}+c^-_{-k,l}+c^-_{k,-l}+c^-_{-k,-l}\\
    B^-_{kl}=& \i (c^-_{k,l}-c^-_{-k,l}+c^-_{k,-l}-c^-_{-k,-l})\\
    C^-_{kl}=& \i (c^-_{k,l}+c_{-k,l}-c^-_{k,-l}-c^-_{-k,-l})\\
    D^-_{kl}=& -c^-_{k,l}+c^-_{-k,l}+c^-_{k,-l}-c^-_{-k,-l}\\
\end{split}
\end{equation}
Substituting equations (\ref{ainterodd}), (\ref{aintereven}) into (\ref{aA}) we obtain the following explicit formulas
for $N=2M+1$
\begin{equation}\label{aAe}
\begin{split}
    A^-_{kl}=&\left(\frac{2}{N}\right)^2\sum_{\setcomb{m,n=0}{m>n}}^{N-1}f(x_m,y_n)  (\cos 2\pi k x_m \cos 2\pi ly_n-\cos 2\pi l x_m \cos 2\pi ky_n )\\
    B^-_{kl}=& \left(\frac{2}{N}\right)^2\sum_{\setcomb{m,n=0}{m>n}}^{N-1}f(x_m,y_n)(\sin 2\pi k x_m \cos 2\pi ly_n-\cos 2\pi l x_m \sin 2\pi ky_n)\\
    C^-_{kl}=& \left(\frac{2}{N}\right)^2\sum_{\setcomb{m,n=0}{m>n}}^{N-1}f(x_m,y_n)(\cos 2\pi k x_m \sin 2\pi ly_n-\sin 2\pi l x_m \cos 2\pi ky_n)\\
    D^-_{kl}=&\left(\frac{2}{N}\right)^2\sum_{\setcomb{m,n=0}{m>n}}^{N-1}f(x_m,y_n) (\sin 2\pi k x_m \sin 2\pi ly_n-\sin 2\pi l x_m \sin 2\pi ky_n)\\
\end{split}
\end{equation}
and for $N=2M$
\begin{equation}\label{aAe2}
\begin{split}
    A^-_{kl}=&g_{k,M}g_{l,M}\left(\frac{2}{N}\right)^2\sum_{\setcomb{m,n=0}{m>n}}^{N-1}f(x_m,y_n)  (\cos 2\pi k x_m \cos 2\pi ly_n-\cos 2\pi l x_m \cos 2\pi ky_n )\\
    B^-_{kl}=& g_{k,M}g_{l,M}\left(\frac{2}{N}\right)^2\sum_{\setcomb{m,n=0}{m>n}}^{N-1}f(x_m,y_n)(\sin 2\pi k x_m \cos 2\pi ly_n-\cos 2\pi l x_m \sin 2\pi ky_n)\\
    C^-_{kl}=& g_{k,M}g_{l,M}\left(\frac{2}{N}\right)^2\sum_{\setcomb{m,n=0}{m>n}}^{N-1}f(x_m,y_n)(\cos 2\pi k x_m \sin 2\pi ly_n-\sin 2\pi l x_m \cos 2\pi ky_n)\\
    D^-_{kl}=&g_{k,M}g_{l,M}\left(\frac{2}{N}\right)^2\sum_{\setcomb{m,n=0}{m>n}}^{N-1}f(x_m,y_n) (\sin 2\pi k x_m \sin 2\pi ly_n-\sin 2\pi l x_m \sin 2\pi ky_n)\\
\end{split}
\end{equation}

\subsubsection{Example of antisymmetric interpolation}\

Consider the following multiple of the Gaussian distribution
\begin{equation}\label{f4}
 f(x,y)=e^{- \frac{(x-x')^2+(y-y')^2}{2\sigma^2}},
\end{equation}
where $(x',y')=(0.707,0.293)$ and $\sigma=0.079$. The function $f$, restricted to the domain $F(S_2^{\mathrm{aff}})$, is depicted in Figure~\ref{Ff}.

\begin{figure}[!ht]
\resizebox{3.2cm}{!}{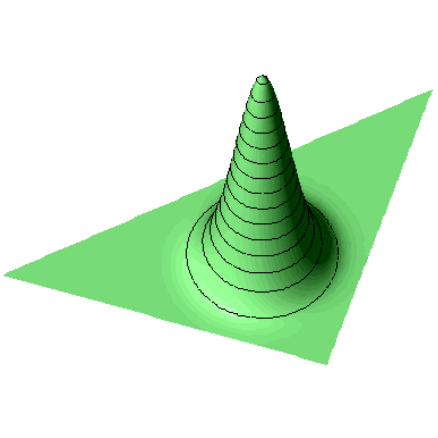}\hspace{1.2cm}
\resizebox{2.4cm}{!}{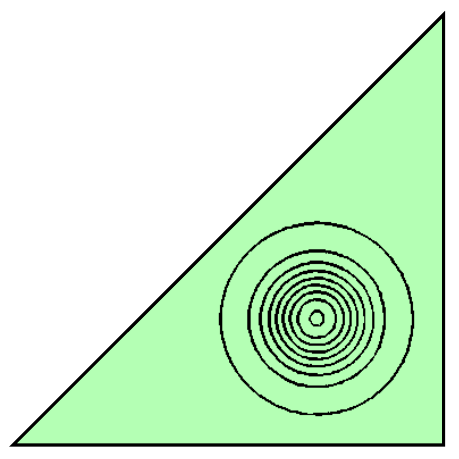}\caption{The function $f$, given by~(\ref{f4}), plotted over the domain $F(S_2^{\mathrm{aff}})$.}\label{Ff}
\end{figure}

We sample the function $f$ on the grids $L^-_{0,\frac{1}{2},4,1}$, $L^-_{0,\frac{1}{2},7,1}$ and $L^-_{0,\frac{1}{2},12,1}$ and calculate the antisymmetric interpolation functions $\psi^-_4$, $\psi^-_7$ and $\psi^-_{12}$ from the relations~(\ref{atrig}) and~(\ref{ainterodd}),~(\ref{aintereven}). These interpolating functions, together with the interpolating grids, are depicted in Figure~\ref{FfEm}. Interpolation errors are summarized in Table~\ref{tabcomp}.

\begin{figure}[!ht]
\resizebox{2.4cm}{!}{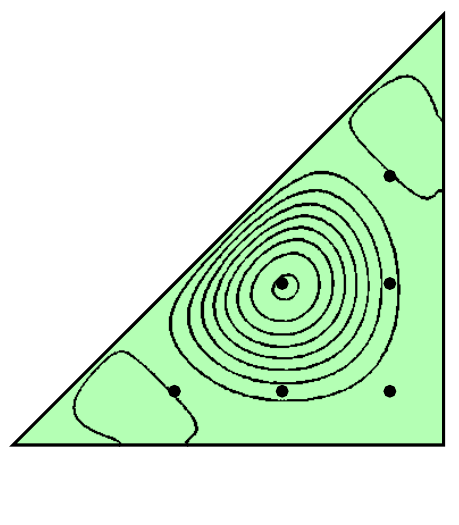}\hspace{1.3cm}
\resizebox{2.4cm}{!}{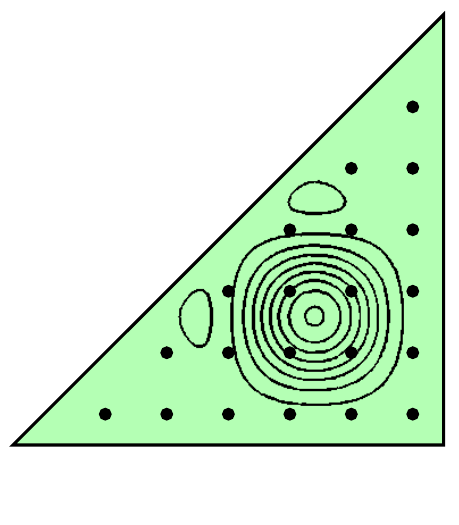}\hspace{1.3cm}
\resizebox{2.4cm}{!}{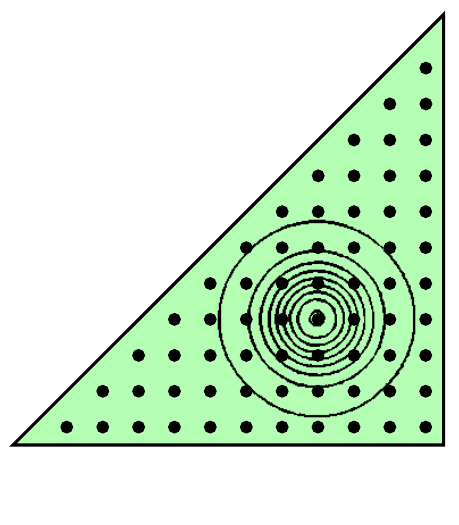}
\\\vspace{1pt}
\resizebox{3.2cm}{!}{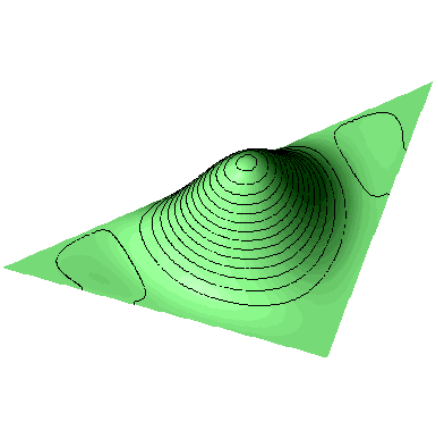}\hspace{14pt}
\resizebox{3.2cm}{!}{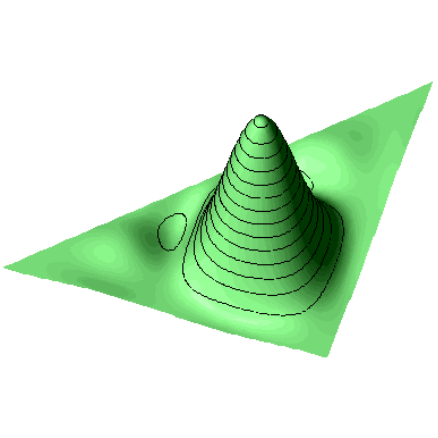}\hspace{14pt}
\resizebox{3.2cm}{!}{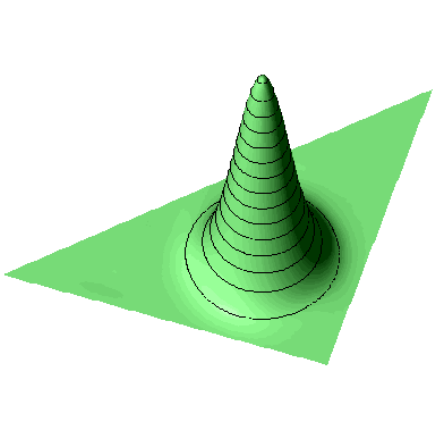}
\caption{The antisymmetric interpolating trigonometric functions $\psi^-_4$, $\psi^-_7$ and $\psi^-_{12}$ of~$f$, given by (\ref{f4}). The points of the interpolation grids $L^-_{0,\frac{1}{2},4,1}$, $L^-_{0,\frac{1}{2},7,1}$ and $L^-_{0,\frac{1}{2},12,1}$ are depicted as small black dots.}\label{FfEm}
\end{figure}

\subsection{Symmetric interpolation}\

For interpolation with symmetric exponential functions, we consider the 'closed' triangle  $K_{[a,a+1]}^+$
inside the square $K_{[a,a']}$:
\begin{equation}\label{sfunds}
K_{[a,a+1]}^+= \set{(x,y)\in[a,a+1]\times[a,a+1]}{
x\geq  y}.
\end{equation}
For a given function $f:K_{[a,a+1]}^+\map \Com$ and a set of points $L^+_{a,b,N,1}\subset K_{[a,a+1]}^+$, we define an {\it symmetric interpolating function}
\begin{equation}\label{strig}
    \psi^+_{N}(x,y)=\sum_{\setcomb{k,l=-M}{k\geq l}}^{M}c^+_{kl}E^+_{(k,l)}(x,y),\q x,y\in \R
\end{equation}
satisfying
\begin{equation}\label{strig2}
    \psi^+_{N}(x_m,y_n)=f(x_m,y_n), \q m\geq n,\,  m,n=0\dots N-1.
\end{equation}

For $N=2M+1$, the antisymmetric interpolating polynomial $\psi^+_{N}$ has $N(N+1)/2$ unknown coefficients $c_{kl}$, to which correspond $N(N+1)/2$ constraints (\ref{strig2}). For $N=2M$, we assume further $2M+1$ conditions
\begin{equation}\label{sassum}
 c^+_{l,-M}= \tau_N c^+_{M,l},\q l=-M\dots M
\end{equation}
and we have $(2M+1)(2M+2)/2-(2M+1)=N(N+1)/2$ unknown coefficients $c^+_{kl}$ corresponding to $N(N+1)/2$ constraints (\ref{strig2}).
For the function $f:K_{[a,a+1]}^+\map \Com$, we define the symmetric extension $Sf:K_{[a,a+1]}\map \Com$ by the formula
\begin{equation}\label{sext}
 Sf(x,y)=\begin{cases}
f(x,y) & x\geq y\\
f(y,x) & x<y.
\end{cases}
\end{equation}
Using symmetric extension, one can prove the following proposition, similarly to Proposition \ref{aip}.
\begin{tvr}
There exists a unique symmetric interpolating polynomial (\ref{strig}) satisfying (\ref{strig2}). The coefficients $c^+_{kl}$ are given for $N=2M+1$ by
\begin{equation}\label{sinterodd}
c^+_{kl}=\frac{1}{G_{kl}N^2}\sum_{\setcomb{m,n=0}{m\geq n}}^{N-1}G_{mn}^{-1}f(x_m,y_n)\overline{E^+_{(k,l)}(x_m,y_n)}
\end{equation}
and for $N=2M$ assuming (\ref{sassum}) by
\begin{equation}\label{sintereven}
 c^+_{kl}=\frac{g_{k,M}g_{l,M}}{G_{kl}N^2}\sum_{\setcomb{m,n=0}{m\geq n}}^{N-1}G_{mn}^{-1}f(x_m,y_n)\overline{E^+_{(k,l)}(x_m,y_n)}.
\end{equation}
\end{tvr}

\subsubsection{Calculation of the coefficients $c^+_{kl}$}\

Instead of the direct calculation of coefficients $c^+_{kl}$, the symmetric discrete Fourier transform (\ref{sbetas}) can be used, and the resulting coefficients $\beta^+_{kl}$ transformed to $c^+_{kl}$'s. By direct comparison of (\ref{sinterodd}), (\ref{sintereven}) to (\ref{sbetas}), we obtain for $N=2M+1$
\begin{equation}
\begin{split}
c^+_{kl}=\beta^+_{kl},\q &k\geq l, k,l=0\dots M \\
c^+_{k,-l}=\tau_N\beta^-_{N-l,k},\q & k=0\dots M,\,l=1\dots M\\
c^-_{-k,-l}=\tau_N^2\beta^-_{N-k,N-l}, \q &k\leq l,\, k,l=1\dots M
\end{split}
\end{equation}
and for $N=2M$
\begin{equation}\label{sbetastoceven}
\begin{split}
c^+_{kl}=g_{k,M}g_{l,M}\beta^+_{kl},\q &k\geq l,\,k,l=0\dots M \\
c^+_{k,-l}=g_{k,M}\tau_N\beta^+_{N-l,k},\q & k=0\dots M,\,l=1\dots M-1\\
c^+_{-k,-l}=\tau_N^2\beta^+_{N-k,N-l}, \q &k\leq l,\, k,l=1\dots M-1
\end{split}
\end{equation}
The formula (\ref{sbetastoceven}) determines $(2M)(2M+1)/2$ coefficients; the rest of the $2M+1$ coefficients are determined via the relations (\ref{aassum}).

\subsubsection{Trigonometric form of $\psi^+_N(x,y)$}\

The symmetric interpolating polynomial (\ref{strig}) can be brought to its trigonometric form:
\begin{equation}\label{sbase}
\begin{split}
    \psi^+_N(x,y)=&\sum_{\setcomb{k,l=0}{k\geq l}}^{M}h_{k}h_{l}G_{kl}^{-1}[A^+_{kl}(\cos 2\pi k x \cos 2\pi ly+\cos 2\pi l x \cos 2\pi ky )\\ +&B^+_{kl}(\sin 2\pi k x \cos 2\pi ly+\cos 2\pi l x \sin 2\pi ky)+\\
                +& C^+_{kl}(\cos 2\pi k x \sin 2\pi ly+\sin 2\pi l x \cos 2\pi ky)\\ +&D^+_{kl}(\sin 2\pi k x \sin 2\pi ly+\sin 2\pi l x \sin 2\pi ky)]
\end{split}
\end{equation}
where
\begin{equation}\label{sA}
\begin{split}
    A^+_{kl}=& c^+_{k,l}+c^+_{-k,l}+c^+_{k,-l}+c^+_{-k,-l}\\
    B^+_{kl}=& \i (c^+_{k,l}-c^+_{-k,l}+c^+_{k,-l}-c^+_{-k,-l})\\
    C^+_{kl}=& \i (c^+_{k,l}+c^+_{-k,l}-c^+_{k,-l}-c^+_{-k,-l})\\
    D^+_{kl}=& -c^+_{k,l}+c^+_{-k,l}+c^+_{k,-l}-c^+_{-k,-l}.\\
\end{split}
\end{equation}
Substituting equations (\ref{ainterodd}), (\ref{aintereven}) into (\ref{aA}) we obtain the following explicit formulas
for $N=2M+1$
\begin{equation}\label{sAe}
\begin{split}
    A^+_{kl}=&\frac{4}{G_{kl}N^2}\sum_{\setcomb{m,n=0}{m\geq n}}^{N-1}G_{mn}^{-1}f(x_m,y_n)  (\cos 2\pi k x_m \cos 2\pi ly_n+\cos 2\pi l x_m \cos 2\pi ky_n )\\
    B^+_{kl}=& \frac{4}{G_{kl}N^2}\sum_{\setcomb{m,n=0}{m\geq n}}^{N-1}G_{mn}^{-1}f(x_m,y_n)(\sin 2\pi k x_m \cos 2\pi ly_n+\cos 2\pi l x_m \sin 2\pi ky_n)\\
    C^+_{kl}=& \frac{4}{G_{kl}N^2}\sum_{\setcomb{m,n=0}{m\geq n}}^{N-1}G_{mn}^{-1}f(x_m,y_n)(\cos 2\pi k x_m \sin 2\pi ly_n+\sin 2\pi l x_m \cos 2\pi ky_n)\\
    D^+_{kl}=&\frac{4}{G_{kl}N^2}\sum_{\setcomb{m,n=0}{m\geq n}}^{N-1}G_{mn}^{-1}f(x_m,y_n) (\sin 2\pi k x_m \sin 2\pi ly_n+\sin 2\pi l x_m \sin 2\pi ky_n)\\
\end{split}
\end{equation}
and for $N=2M$
\begin{equation}\label{sAe2}
\begin{split}
    A^+_{kl}=&\frac{4g_{k,M}g_{l,M}}{G_{kl}N^2}\sum_{\setcomb{m,n=0}{m\geq n}}^{N-1}G_{mn}^{-1}f(x_m,y_n)  (\cos 2\pi k x_m \cos 2\pi ly_n+\cos 2\pi l x_m \cos 2\pi ky_n )\\
    B^+_{kl}=& \frac{4g_{k,M}g_{l,M}}{G_{kl}N^2}\sum_{\setcomb{m,n=0}{m\geq n}}^{N-1}G_{mn}^{-1}f(x_m,y_n)(\sin 2\pi k x_m \cos 2\pi ly_n+\cos 2\pi l x_m \sin 2\pi ky_n)\\
    C^+_{kl}=& \frac{4g_{k,M}g_{l,M}}{G_{kl}N^2}\sum_{\setcomb{m,n=0}{m\geq n}}^{N-1}G_{mn}^{-1}f(x_m,y_n)(\cos 2\pi k x_m \sin 2\pi ly_n+\sin 2\pi l x_m \cos 2\pi ky_n)\\
    D^+_{kl}=&\frac{4g_{k,M}g_{l,M}}{G_{kl}N^2}\sum_{\setcomb{m,n=0}{m\geq n}}^{N-1}G_{mn}^{-1}f(x_m,y_n) (\sin 2\pi k x_m \sin 2\pi ly_n+\sin 2\pi l x_m \sin 2\pi ky_n)\\
\end{split}
\end{equation}

\subsubsection{Example of symmetric interpolation}\

We sample the function $f$, given by (\ref{f4}), on the grids $L^+_{0,\frac{1}{2},4,1}$, $L^+_{0,\frac{1}{2},7,1}$ and $L^+_{0,\frac{1}{2},12,1}$, and calculate the symmetric interpolation functions $\psi^+_4$, $\psi^+_7$ and $\psi^+_{12}$ from the relations~(\ref{strig}) and~(\ref{sinterodd}),~(\ref{sintereven}). These interpolating functions, together with the interpolating grids, are depicted in Figure~\ref{FfEp}. Interpolations errors are summarized in Table~\ref{tabcomp}.

\begin{figure}[!ht]
\resizebox{2.4cm}{!}{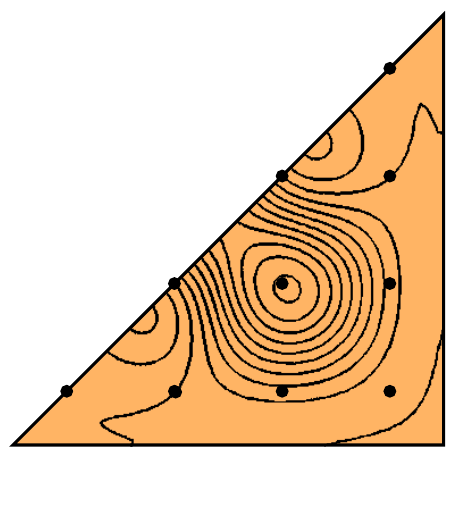}\hspace{1.3cm}
\resizebox{2.4cm}{!}{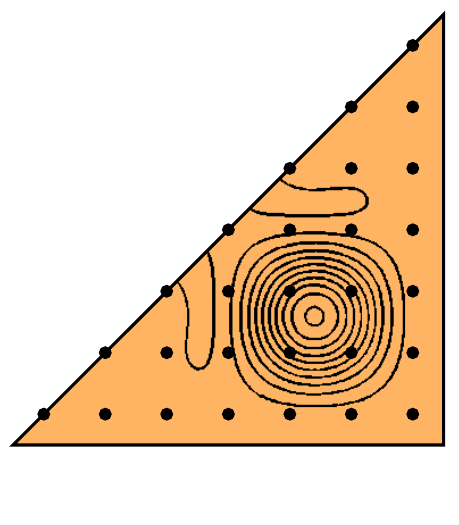}\hspace{1.3cm}
\resizebox{2.4cm}{!}{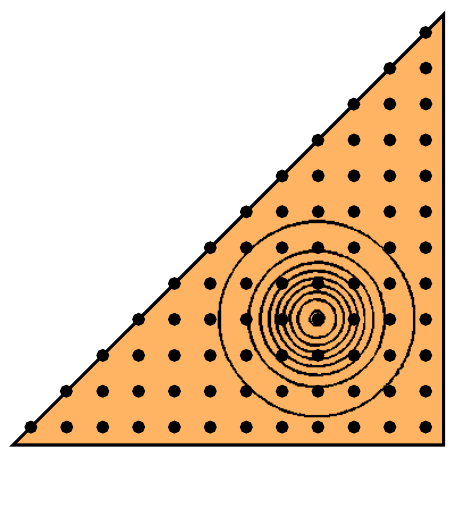}
\\\vspace{1pt}
\resizebox{3.2cm}{!}{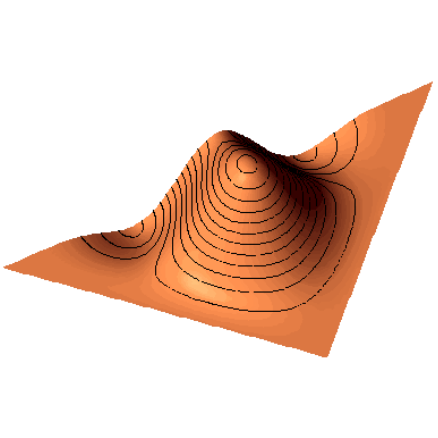}\hspace{14pt}
\resizebox{3.2cm}{!}{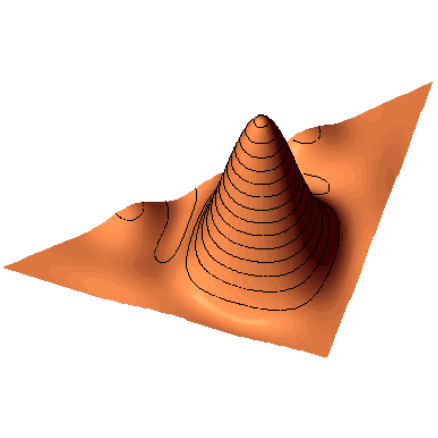}\hspace{14pt}
\resizebox{3.2cm}{!}{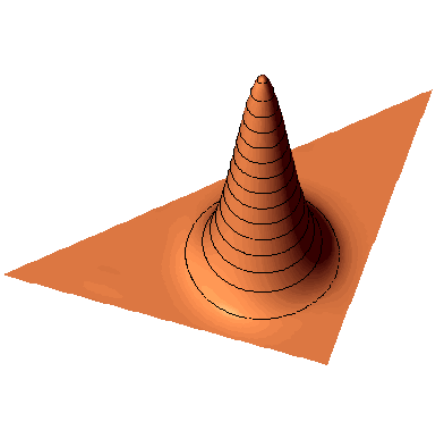}
\caption{The symmetric interpolating trigonometric functions $\psi^+_4$, $\psi^+_7$ and $\psi^+_{12}$ of~$f$, given by (\ref{f4}). The points of the interpolation grids $L^+_{0,\frac{1}{2},4,1}$, $L^+_{0,\frac{1}{2},7,1}$ and $L^+_{0,\frac{1}{2},12,1}$ are depicted as small black dots.}\label{FfEp}
\end{figure}

\section{(Anti)symmetric cosine transforms}

\subsection{Antisymmetric cosine functions}

\subsubsection{Definitions, symmetries and general properties}\

Two-dimensional antisymmetric cosine functions ${\cos}^-_{(\lambda,\mu)}:\R^2\map \Com$ have for $\lambda,\mu\in \R$ the following explicit form
\begin{equation*}
{\cos}^-_{(\lambda,\mu)}(x,y)
     =\left|\begin{smallmatrix}
     \cos(\pi\lambda x)&\cos(\pi\lambda y)\\
     \cos(\pi\mu x)&\cos(\pi\mu y)\\
     \end{smallmatrix}\right|
     =\cos(\pi\lambda x)\cos(\pi\mu y)
     -\cos(\pi\mu x)\cos(\pi\lambda y)
\end{equation*}
Note that we chose this definition according to Section X of~\cite{KPtrig}. Instead of the factor $2\pi$ we use 'half' argument $\pi$.We observe that ${\cos}^-_{(\lambda,\lambda)}(x,y)=0$ and ${\cos}^-_{(\lambda,\mu)}(x,x)=0$. From the explicit formula, we immediately obtain antisymmetry of ${\cos}^-_{(\lambda,\mu)}(x,y)$ with respect to the permutation of variables~$(x,y)$
\begin{equation}\label{cosant}
 {\cos}^-_{(\lambda,\mu)}(y,x)=-{\cos}^-_{(\lambda,\mu)}(x,y).
\end{equation}
and, moreover, with respect to the permutation of $(\la,\mu)$
\begin{equation}\label{cosant2}
{\cos}^-_{(\lambda,\mu)}(x,y)= -{\cos}^-_{(\mu,\lambda)}(x,y).
\end{equation}
Therefore, we consider only such ${\cos}^-_{(\lambda,\mu)}$ with strictly dominant $(\lambda,\mu)$, $\la>\mu$. The functions ${\cos}^-_{(k,l)}$ with $k,l\in\Z$ have symmetries related to the periodicity of cosine function
\begin{equation}\label{acosper}
    {\cos}^-_{(k,l)}(x+2r,y+2s)= {\cos}^-_{(k,l)}(x,y),\q r,s\in \Z
\end{equation}
We also have invariance under the change of sign of $(x,y)$
\begin{equation}\label{acossign}
    {\cos}^-_{(k,l)}(-x,y)= {\cos}^-_{(k,l)}(x,-y)={\cos}^-_{(k,l)}(-x,-y)={\cos}^-_{(k,l)}(x,y)
\end{equation}
and invariance under the change of sign of $(k,l)$
\begin{equation}\label{klacossign}
    {\cos}^-_{(-k,l)}(x,y)= {\cos}^-_{(k,-l)}(x,y)={\cos}^-_{(-k,-l)}(x,y)={\cos}^-_{(k,l)}(x,y).
\end{equation}
The relations (\ref{cosant}) -- (\ref{klacossign}) imply that it is sufficient to consider the functions ${\cos}^-_{(k,l)},\,k,l\in\Z^{\geq 0},\,k>l$ on the fundamental domain $F(S_2^{\mathrm{aff}})$.

\subsubsection{Continuous orthogonality}\

The functions ${\cos}^-_{(k,l)}$ are mutually orthogonal on the fundamental domain $F(S_2^{\mathrm{aff}})$, i.e.,
\begin{equation*}
    \int_{F(S_2^{\mathrm{aff}})} {\cos}^-_{(k,l)}(x,y)
\overline{{\cos}^-_{(k',l')}(x,y)}\,dx\, dy.=
\frac{1}{4}\delta_{kk'}\delta_{ll'}, \q k,l,k',l' \in\Z^{\geq 0},\, k>l,k'>l'.  \end{equation*}
Also every function $f:\R^2\map \Com$ that is antisymmetric $f(x,y)=-f(y,x)$ and periodic $f(x+2r,y+2s)= f(x,y),\, r,s\in \Z$ and has continuous derivatives can be expanded in the antisymmetric cosine functions ${\cos}^-_{(k,l)}$:
\begin{equation}
f(x,y)=\sum_{\setcomb{k,l\in \Z^{\geq 0}}{k>l}} {\wt c}_{kl} {\cos}^-_{(k,l)}(x,y)
,\q
{\wt c}_{kl} = 4 \int_{F(S^{\mathrm{aff}}_2)} f(x,y)
{\cos}^-_{(k,l)}(x,y)\,dx\, dy.
 \end{equation}

The graphs of the lowest antisymmetric cosine functions ${\cos}^-_{(k,l)}$,
$k,l\in\{0,\dots,3\},\,k>l $ are plotted in Figure~\ref{FCm}.

\begin{figure}[!ht]
\resizebox{2.4cm}{!}{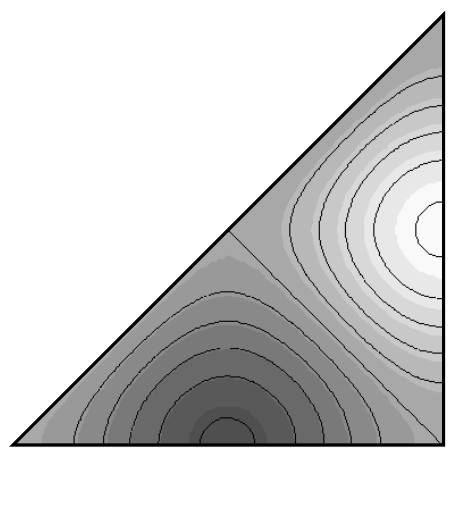}\hspace{22pt}
\resizebox{2.4cm}{!}{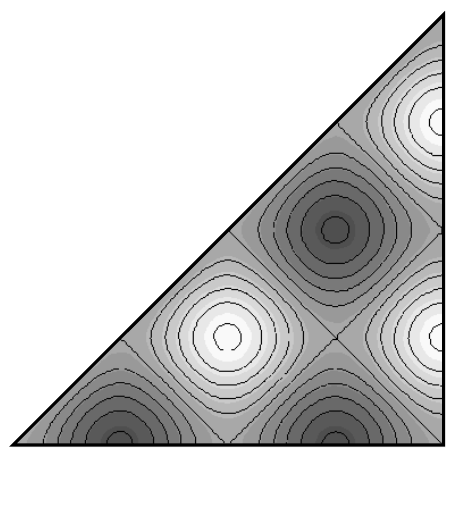}\hspace{22pt}
\resizebox{2.4cm}{!}{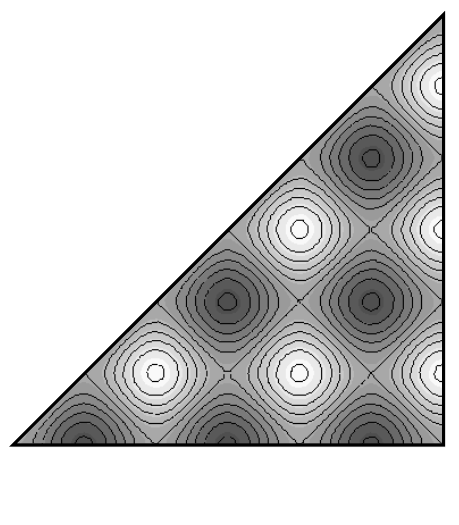}
\\\vspace{2pt}
\resizebox{2.4cm}{!}{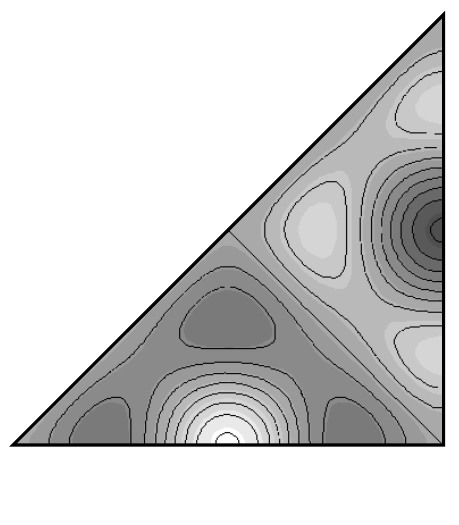}\hspace{22pt}
\resizebox{2.4cm}{!}{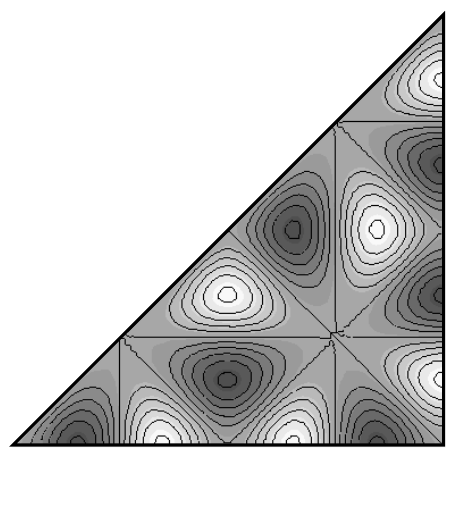}\hspace{22pt}
\resizebox{2.4cm}{!}{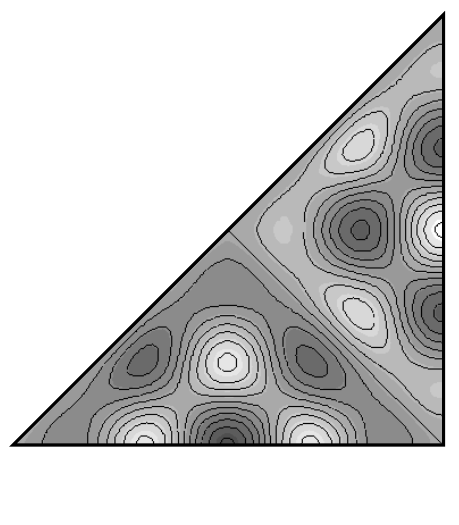} \\ \caption{The contour plots of the antisymmetric cosine functions ${\cos}^-_{(k,l)}$,
$k,l\in\{0,\dots,3\},\,k>l. $}\label{FCm}
\end{figure}

\subsubsection{Solutions of the Laplace equation}\

The functions $\cos^-$ are solutions of the Laplace equation
\begin{equation}\label{laplaccm}
\left(\frac{\partial^2}{\partial x^2}+\frac{\partial^2}{\partial y^2}\right)\cos^-_{(k,l)}(x,y)=-\pi^2(k^2+l^2)\cos^-_{(k,l)}(x,y)
\end{equation}
and moreover of the equation
\begin{equation*}
\frac{\partial^2}{\partial x^2}\frac{\partial^2}{\partial y^2}\cos^-_{(k,l)}(x,y)=\pi^4k^2l^2\cos^-_{(k,l)}(x,y).
\end{equation*}
The functions $\cos^-$ satisfy the condition $\cos^-_{(k,l)}(x,y)=0$ on the boundary $x=y$ and the condition $$\frac{\partial}{\partial \mathbf{n} }\cos^-_{(k,l)}(x,y)=0,$$ where $\mathbf{n}$ is the normal to the boundary $x=1$ or $y=0$.

\subsection{Antisymmetric discrete cosine transforms}\

Four types of discrete antisymmetric cosine transforms can be derived from the antisymmetric trigonometric transform. In order to derive these transforms, we define the following two functional operators. To a function $f:K_{[0,L]}  \map\Com$ we assign a function $E_L f:K_{[-L,L]} \map \Com $ defined by \begin{equation}\label{extE}
          E_Lf(x,y):=  \begin{cases} f(x,y) & x\geq 0,y\geq0 \\ f(-x,y) & x<0,y>0 \\ f(x,-y) & x>0,y<0 \\ f(-x,-y) & x\leq 0,y\leq 0 ,\end{cases}
           \end{equation}

and to a function $f:K_{[0,1]} \map \Com $, we assign a function $Rf:K_{[0,2]} \map \Com $ defined by
\begin{equation}
          Rf(x,y):=  \begin{cases} f(x,y) & 0 \leq x \leq 1,0 \leq y \leq 1 \\ -f(2-x,y) & 1 < x \leq 2,0 \leq y \leq 1 \\ -f(x,2-y) & 0 \leq x \leq 1,1< y \leq 2 \\ f(2-x,2-y) & 1 < x \leq 2,1 < y \leq 2 .\end{cases}
\end{equation}

All four antisymmetric cosine transforms operate on a function $f:K_{[0,1]}^-\map \Com$. Applying the formula (\ref{aAe2})
to the four functions
\begin{enumerate}[(I)]
 \item  $E_1Af:K_{[-1,1]} \map \Com $, where $N=2M$, $T=2$, $b=1$
 \item  $E_1Af:K_{[-1,1]} \map \Com $, where $N=2M$, $T=2$, $b=1/2$
 \item  $E_2RA\, f:K_{[-2,2]} \map \Com $, where $N=4M$, $T=4$, $b=1$
 \item  $E_2RA\, f:K_{[-2,2]} \map \Com $, where $N=4M$, $T=4$, $b=1/2$
\end{enumerate}
we obtain due to (anti)symmetry in (\ref{base}) that $B^-_{kl}=C^-_{kl}=D^-_{kl}=0$. Introducing the symbol $d_{k,M}$
for $k=0,\dots,M$ by \begin{equation}
d_{k,M}=\begin{cases}\frac12 & \text{if $k=0,M$} \\ 1 & \text{otherwise} \end{cases}
\end{equation}
the four interpolating functions (\ref{base}) corresponding to cases (I)--(IV) and their coefficients (\ref{aAe2}) can be brought to the following form.

\begin{enumerate}[{AMDCT-} I.]
 \item  $$\psi^{\mathrm{I},-}_M(x,y)=\sum_{\setcomb{k,l=0}{k>l}}^{M} c_{k,l}^{\mathrm{I},-}{\cos}^-_{(k,l)}(x,y),\q c_{k,l}^{\mathrm{I},-}=\frac{4d_{k,M} d_{l,M}}{M^2} \sum_{\setcomb{m,n=0}{m>n}}^{M}d_{m,M} d_{n,M}f\left(x_m,y_n\right){\cos}^-_{(k,l)}\left(x_m,y_n\right)$$
where $x_m=\frac{m}{M}$, $y_n=\frac{n}{M}$.
 \item

$$\psi^{\mathrm{II},-}_M(x,y)=\sum_{\setcomb{k,l=0}{k>l}}^{M-1} c_{k,l}^{\mathrm{II},-}{\cos}^-_{(k,l)}(x,y),\q c_{k,l}^{\mathrm{II},-}=\frac{4d_{k,M} d_{l,M}}{M^2} \sum_{\setcomb{m,n=0}{m>n}}^{M-1}f\left(x_m,y_n\right){\cos}^-_{(k,l)}\left(x_m,y_n\right)$$
where $x_m=\frac{m+\frac12}{M}$, $y_n=\frac{n+\frac12}{M}$.

 \item  $$\psi^{\mathrm{III},-}_M(x,y)=\sum_{\setcomb{k,l=0}{k>l}}^{M-1} c_{k,l}^{\mathrm{III},-}{\cos}^-_{(k+\frac12,l+\frac12)}(x,y),$$ $$  c_{k,l}^{\mathrm{III},-}=\frac{4}{M^2} \sum_{\setcomb{m,n=0}{m>n}}^{M-1}d_{m,M}d_{n,M}f\left(x_m,y_n\right){\cos}^-_{(k+\frac12,l+\frac12)}\left(x_m,y_n\right) $$
where $x_m=\frac{m}{M}$, $y_n=\frac{n}{M}$.
 \item  $$\psi^{\mathrm{IV},-}_M(x,y)=\sum_{\setcomb{k,l=0}{k>l}}^{M-1} c_{kl}^{\mathrm{IV},-}{\cos}^-_{(k+\frac12,l+\frac12)}(x,y),\q c_{kl}^{\mathrm{IV},-}=\frac{4}{M^2} \sum_{\setcomb{m,n=0}{m>n}}^{M-1}f\left(x_m,y_n\right){\cos}^-_{(k+\frac12,l+\frac12)}\left(x_m,y_n\right)$$
where $x_m=\frac{m+\frac12}{M}$, $y_n=\frac{n+\frac12}{M}$.
\end{enumerate}

\subsubsection{Example of antisymmetric cosine interpolation}\

We sample the function $f$ on the grids $L^-_{0,\frac{1}{2},4,1}$, $L^-_{0,\frac{1}{2},7,1}$ and $L^-_{0,\frac{1}{2},12,1}$ and calculate the antisymmetric cosine interpolating function of the type AMDCT-II: $\psi^{\mathrm{II},-}_4$, $\psi^{\mathrm{II},-}_7$ and $\psi^{\mathrm{II},-}_{12}$. These interpolating functions, together with the interpolating grids, are depicted in Figure~\ref{FfCm}. Interpolation errors are summarized in Table~\ref{tabcomp}.

\begin{figure}[!ht]
\resizebox{2.4cm}{!}{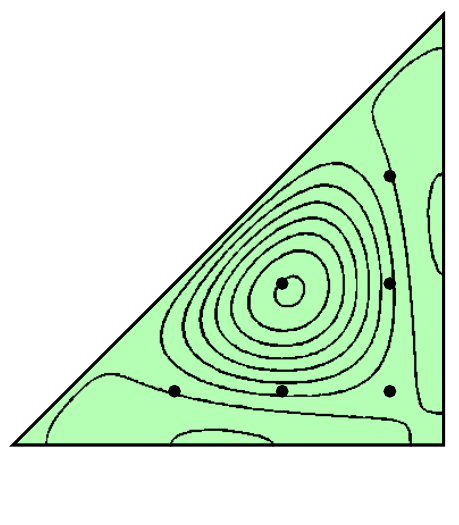}\hspace{1.3cm}
\resizebox{2.4cm}{!}{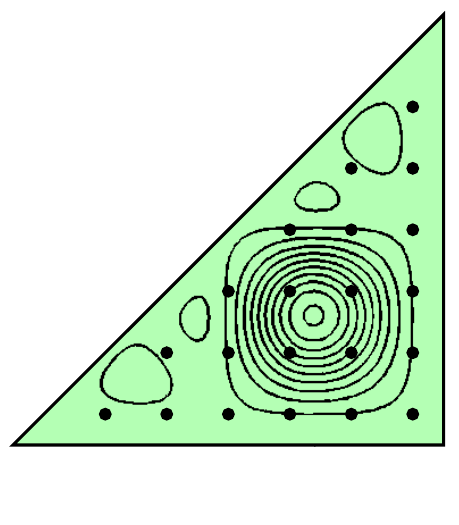}\hspace{1.3cm}
\resizebox{2.4cm}{!}{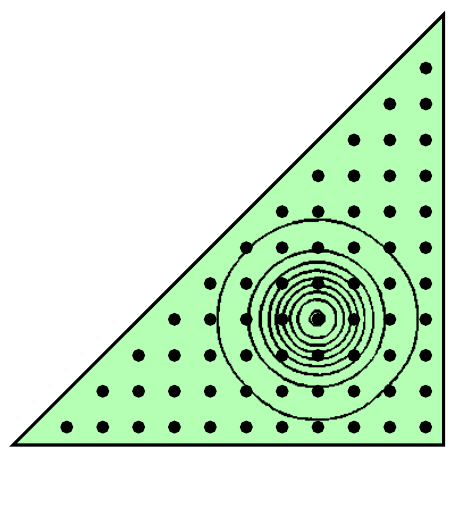}
\\\vspace{1pt}
\resizebox{3.2cm}{!}{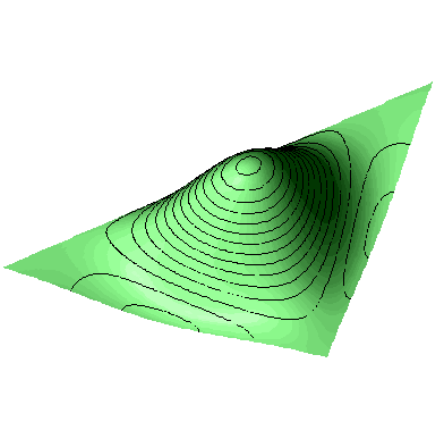}\hspace{14pt}
\resizebox{3.2cm}{!}{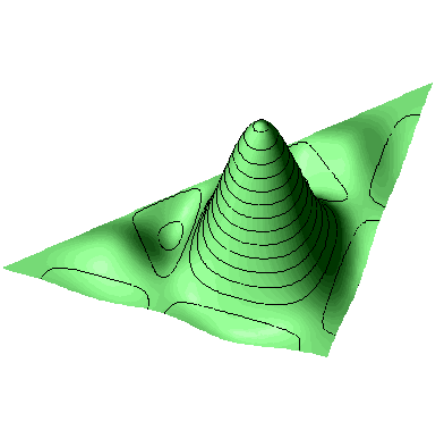}\hspace{14pt}
\resizebox{3.2cm}{!}{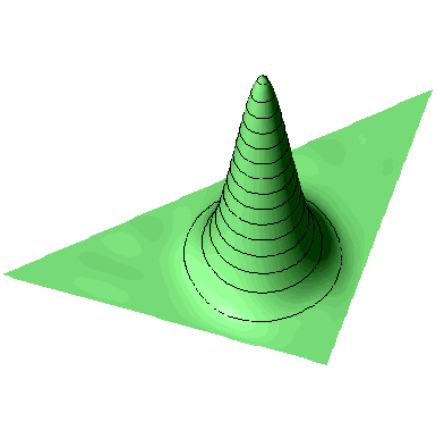}
\caption{The antisymmetric interpolating cosine functions of the type AMDCT-II: $\psi^{\mathrm{II},-}_4$, $\psi^{\mathrm{II},-}_7$ and $\psi^{\mathrm{II},-}_{12}$ of the function $f$, given by (\ref{f4}). The points of the interpolation grids $L^-_{0,\frac{1}{2},4,1}$, $L^-_{0,\frac{1}{2},7,1}$ and $L^-_{0,\frac{1}{2},12,1}$ are depicted as small black dots.}\label{FfCm}
\end{figure}

\subsection{Symmetric cosine functions}

\subsubsection{Definitions, symmetries and general properties}\

Two-dimensional symmetric cosine functions ${\cos}^+_{(\lambda,\mu)}:\R^2\map \Com$ have for $\lambda,\mu\in \R$ the following explicit form
\begin{equation*}
{\cos}^+_{(\lambda,\mu)}(x,y)
     =\left|\begin{smallmatrix}
     \cos(\pi\lambda x)&\cos(\pi\lambda y)\\
     \cos(\pi\mu x)&\cos(\pi\mu y)\\
     \end{smallmatrix}\right|^+
     =\cos(\pi\lambda x)\cos(\pi\mu y)
     +\cos(\pi\mu x)\cos(\pi\lambda y)
\end{equation*}
Note that we chose this definition according to Section X of~\cite{KPtrig}. Instead of the factor $2\pi$ we use 'half' argument $\pi$.
From the explicit formula we immediately obtain symmetry of ${\cos}^+_{(\lambda,\mu)}(x,y)$ with respect to the permutation of variables~$(x,y)$
\begin{equation}\label{cossym}
 {\cos}^+_{(\lambda,\mu)}(y,x)={\cos}^+_{(\lambda,\mu)}(x,y).
\end{equation}
and, moreover, with respect to the permutation of $(\la,\mu)$
\begin{equation}\label{cossym2}
{\cos}^+_{(\lambda,\mu)}(x,y)= {\cos}^+_{(\mu,\lambda)}(x,y).
\end{equation}
Therefore, we consider only such ${\cos}^+_{(\lambda,\mu)}$ with dominant $(\lambda,\mu)$, $\la\geq\mu$. The functions ${\cos}^+_{(k,l)}$ with $k,l\in\Z$ have symmetries related to the periodicity of cosine function
\begin{equation}\label{scosper}
    {\cos}^+_{(k,l)}(x+2r,y+2s)= {\cos}^+_{(k,l)}(x,y),\q r,s\in \Z.
\end{equation}
We also have invariance under the change of sign of variables $(x,y)$
\begin{equation}\label{scossign}
    {\cos}^+_{(k,l)}(-x,y)= {\cos}^+_{(k,l)}(x,-y)={\cos}^+_{(k,l)}(-x,-y)={\cos}^+_{(k,l)}(x,y)
\end{equation}
and under the change of sign of $(k,l)$
\begin{equation}\label{klscossign}
    {\cos}^+_{(-k,l)}(x,y)= {\cos}^+_{(k,-l)}(x,y)={\cos}^+_{(-k,-l)}(x,y)={\cos}^+_{(k,l)}(x,y).
\end{equation}

The relations (\ref{cossym}) -- (\ref{klscossign}) imply that it is sufficient to consider the functions ${\cos}^+_{(k,l)},\,k,l\in\Z^{\geq 0},\,k\geq l$ on the fundamental domain $F(S_2^{\mathrm{aff}})$.

\subsubsection{Continuous orthogonality}\

The functions ${\cos}^+_{(k,l)}$ are mutually orthogonal on the fundamental domain $F(S_2^{\mathrm{aff}})$, i.e.,
\begin{equation*}
    \int_{F(S_2^{\mathrm{aff}})} {\cos}^+_{(k,l)}(x,y)
\overline{{\cos}^+_{(k',l')}(x,y)}\,dx\, dy.=
\frac{G_{kl}}{4}\delta_{kk'}\delta_{ll'}, \q k,l,k',l' \in\Z^{\geq 0},\, k\geq l,k' \geq l'.  \end{equation*}
Also every function $f:\R^2\map \Com$ that is symmetric $f(x,y)=f(y,x)$ and periodic $f(x+2r,y+2s)= f(x,y),\, r,s\in \Z$ and has continuous derivatives can be expanded in the symmetric cosine functions ${\cos}^+_{(k,l)}$:
\begin{equation}
f(x,y)=\sum_{\setcomb{k,l\in \Z^{\geq 0}}{k\geq l}} {\wt c}_{kl} {\cos}^+_{(k,l)}(x,y)
,\q
{\wt c}_{kl} = 4G_{kl}^{-1} \int_{F(S^{\mathrm{aff}}_2)} f(x,y)
{\cos}^+_{(k,l)}(x,y)\,dx\, dy.
 \end{equation}
The graphs of the lowest symmetric cosine functions ${\cos}^+_{(k,l)}$,
$k,l\in\{0,1,2\},\,k\geq l $ are plotted in Figure~\ref{FCp}.
\begin{figure}[!ht]
\resizebox{2.4cm}{!}{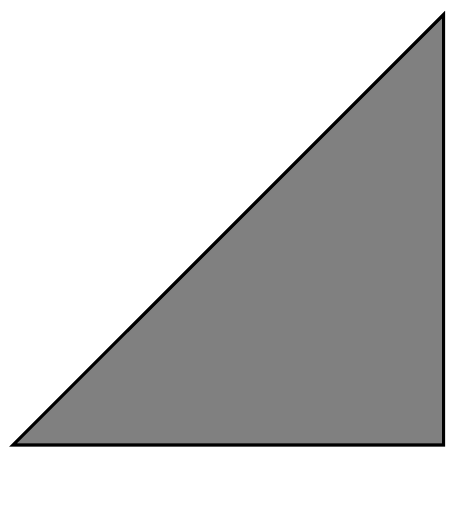}\hspace{22pt}
\resizebox{2.4cm}{!}{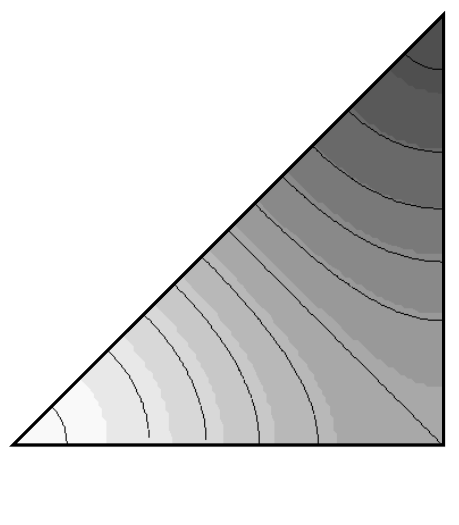}\hspace{22pt}
\resizebox{2.4cm}{!}{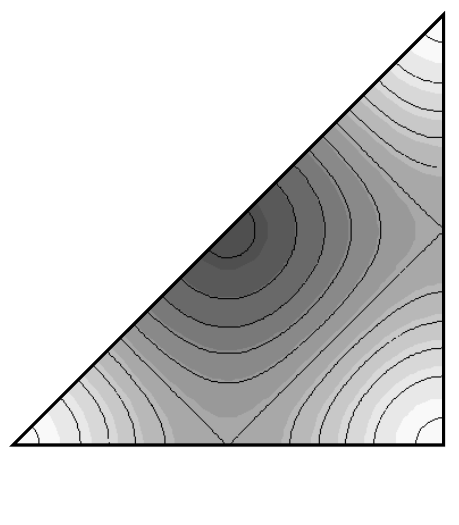}
\\\vspace{2pt}
\resizebox{2.4cm}{!}{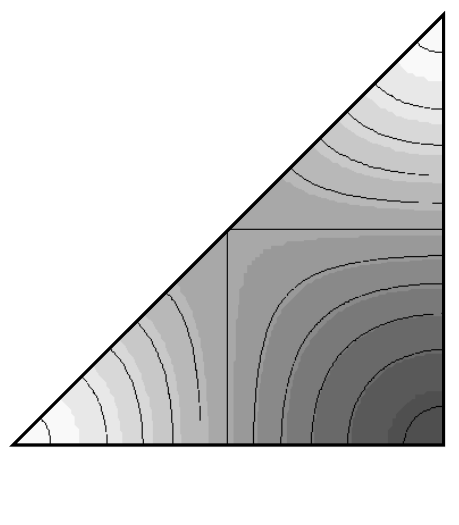}\hspace{22pt}
\resizebox{2.4cm}{!}{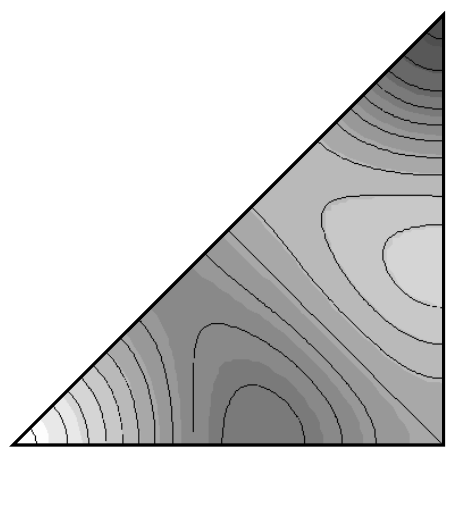}\hspace{22pt}
\resizebox{2.4cm}{!}{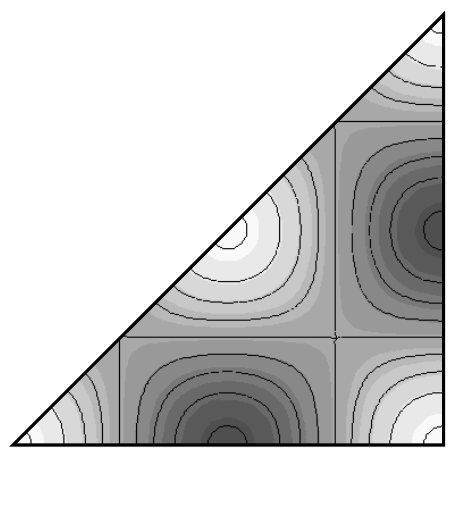} \\ \caption{The contour plots of the symmetric cosine functions ${\cos}^+_{(k,l)}$,
$k,l\in\{0,1,2\},\,k\geq l. $}\label{FCp}
\end{figure}

\subsubsection{Solutions of the Laplace equation}\

The functions $\cos^+$ are solutions of the Laplace equation
\begin{equation*}
\left(\frac{\partial^2}{\partial x^2}+\frac{\partial^2}{\partial y^2}\right)\cos^+_{(k,l)}(x,y)=-\pi^2(k^2+l^2)\cos^+_{(k,l)}(x,y)
\end{equation*}
and of the equation
\begin{equation*}
\frac{\partial^2}{\partial x^2}\frac{\partial^2}{\partial y^2}\cos^+_{(k,l)}(x,y)=\pi^4k^2l^2\cos^+_{(k,l)}(x,y).
\end{equation*}
The functions $\cos^+$ satisfy the condition $$\frac{\partial}{\partial \mathbf{n} }\cos^+_{(k,l)}(x,y)=0,$$ where $\mathbf{n}$ is the normal to the boundary of $F(S^{\mathrm{aff}}_2)$, i.e. are the solutions of the Neumann boundary value problem for $F(S^{\mathrm{aff}}_2)$.

\subsection{Symmetric discrete cosine transforms}\

Four types of discrete symmetric cosine transforms can be derived from the symmetric trigonometric transform.
All four symmetric cosine transforms operate on a function $f:K_{[0,1]}^+\map \Com$. Applying the formula (\ref{sAe2})
to the four functions
\begin{enumerate}[(I)]
 \item  $E_1Sf:K_{[-1,1]} \map \Com $, where $N=2M$, $T=2$, $b=1$
 \item  $E_1Sf:K_{[-1,1]} \map \Com $, where $N=2M$, $T=2$, $b=1/2$
 \item  $E_2RS\, f:K_{[-2,2]} \map \Com $, where $N=4M$, $T=4$, $b=1$
 \item  $E_2RS\, f:K_{[-2,2]} \map \Com $, where $N=4M$, $T=4$, $b=1/2$
\end{enumerate}
we obtain due to symmetry in (\ref{sbase}) that $B^+_{kl}=C^+_{kl}=D^+_{kl}=0$.
The four interpolating functions (\ref{sbase}) corresponding to cases (I)--(IV) and their coefficients (\ref{sAe2}) can be brought to the following form.

\begin{enumerate}[{SMDCT--}I.]
 \item  $$\psi^{\mathrm{I},+}_M(x,y)=\sum_{\setcomb{k,l=0}{k\geq l}}^{M} c_{kl}^{\mathrm{I},+}{\cos}^+_{(k,l)}(x,y),$$ $$ c_{kl}^{\mathrm{I},+}=\frac{4d_{k,M} d_{l,M}}{M^2G_{kl}} \sum_{\setcomb{m,n=0}{m\geq n}}^{M}G_{mn}^{-1}d_{m,M} d_{n,M}f\left(x_m,y_n\right){\cos}^+_{(k,l)}\left(x_m,y_n\right)$$
where $x_m=\frac{m}{M}$, $y_n=\frac{n}{M}$.
 \item

$$\psi^{\mathrm{II},+}_M(x,y)=\sum_{\setcomb{k,l=0}{k\geq l}}^{M-1} c_{k,l}^{\mathrm{II},+}{\cos}^+_{(k,l)}(x,y),\q c_{k,l}^{\mathrm{II},+}=\frac{4d_{k,M} d_{l,M}}{M^2G_{kl}} \sum_{\setcomb{m,n=0}{m\geq n}}^{M-1}G_{mn}^{-1}f\left(x_m,y_n\right){\cos}^+_{(k,l)}\left(x_m,y_n\right)$$
where $x_m=\frac{m+\frac12}{M}$, $y_n=\frac{n+\frac12}{M}$.

 \item  $$\psi^{\mathrm{III},+}_M(x,y)=\sum_{\setcomb{k,l=0}{k\geq l}}^{M-1} c_{k,l}^{\mathrm{III},+}{\cos}^+_{(k+\frac12,l+\frac12)}(x,y),$$ $$  c_{k,l}^{\mathrm{III},+}=\frac{4}{M^2G_{kl}} \sum_{\setcomb{m,n=0}{m\geq n}}^{M-1}d_{m,M}d_{n,M}G_{mn}^{-1}f\left(x_m,y_n\right){\cos}^+_{(k+\frac12,l+\frac12)}\left(x_m,y_n\right) $$
where $x_m=\frac{m}{M}$, $y_n=\frac{n}{M}$.
 \item  $$\psi^{\mathrm{IV},+}_M(x,y)=\sum_{\setcomb{k,l=0}{k\geq l}}^{M-1} c_{kl}^{\mathrm{IV},+}{\cos}^+_{(k+\frac12,l+\frac12)}(x,y),$$ $$ c_{kl}^{\mathrm{IV},+}=\frac{4}{M^2G_{kl}} \sum_{\setcomb{m,n=0}{m\geq n}}^{M-1}G_{mn}^{-1}f\left(x_m,y_n\right){\cos}^+_{(k+\frac12,l+\frac12)}\left(x_m,y_n\right)$$
where $x_m=\frac{m+\frac12}{M}$, $y_n=\frac{n+\frac12}{M}$.
\end{enumerate}

\subsubsection{Two Examples}\

\begin{enumerate}[(1)]
\item
We sample the function $f$ on the grids $L^+_{0,\frac{1}{2},4,1}$, $L^+_{0,\frac{1}{2},7,1}$ and $L^+_{0,\frac{1}{2},12,1}$ and calculate the symmetric cosine interpolating function of type SMDCT-II: $\psi^{\mathrm{II},+}_4$, $\psi^{\mathrm{II},+}_7$ and $\psi^{\mathrm{II},+}_{12}$. These interpolating functions, together with the interpolating grids, are depicted in Figure~\ref{FfCp}. Interpolation errors are summarized in Table~\ref{tabcomp}.

\begin{figure}[!ht]
\resizebox{2.4cm}{!}{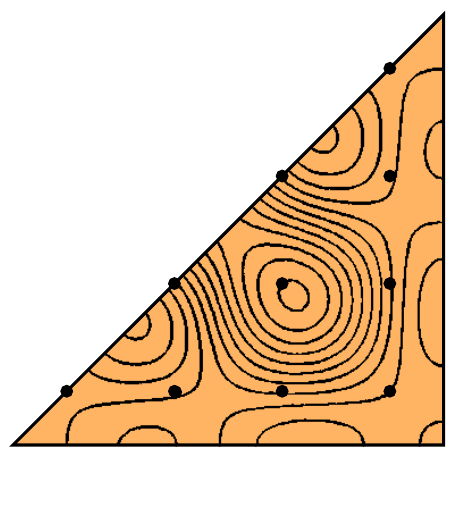}\hspace{1.3cm}
\resizebox{2.4cm}{!}{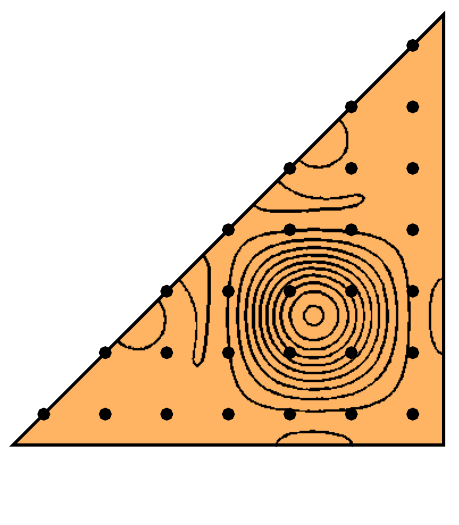}\hspace{1.3cm}
\resizebox{2.4cm}{!}{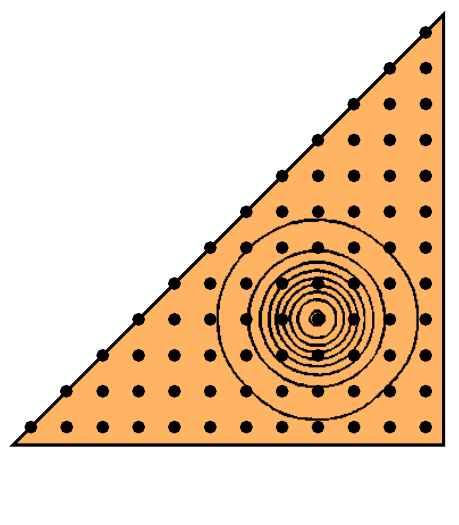}
\\\vspace{1pt}
\resizebox{3.2cm}{!}{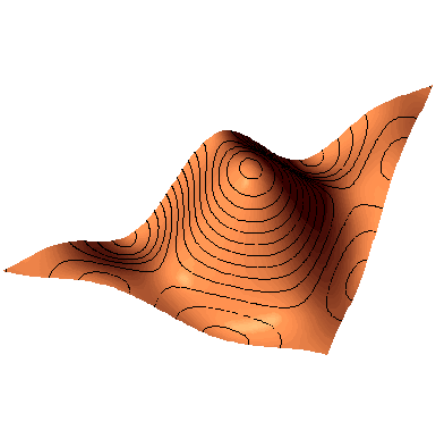}\hspace{14pt}
\resizebox{3.2cm}{!}{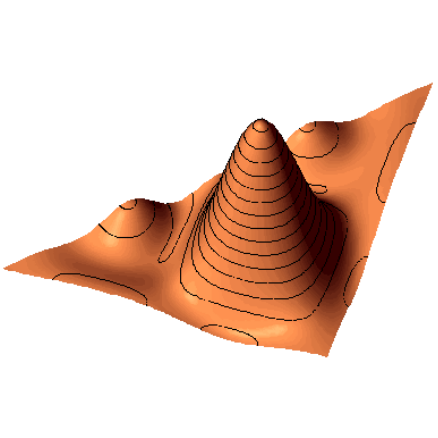}\hspace{14pt}
\resizebox{3.2cm}{!}{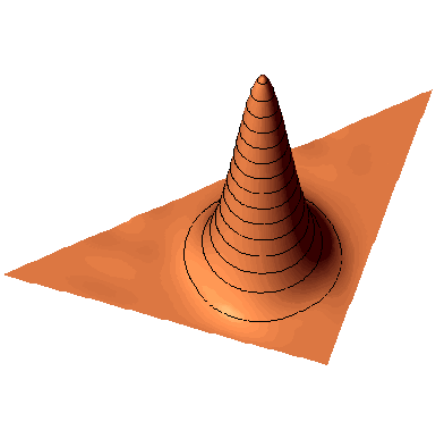}
\caption{The symmetric interpolating cosine functions of type SMDCT-II: $\psi^{\mathrm{II},+}_4$, $\psi^{\mathrm{II},+}_7$ and $\psi^{\mathrm{II},+}_{12}$ of~$f$, given by (\ref{f4}). The points of the interpolation grids $L^+_{0,\frac{1}{2},4,1}$, $L^+_{0,\frac{1}{2},7,1}$ and $L^+_{0,\frac{1}{2},12,1}$ are depicted as small black dots.}\label{FfCp}
\end{figure}

\begin{table}
\begin{tabular}{|c||r||r||r||r|}
\hline
 $N$ & $ \int_{F(S_2^{\mathrm{aff}})} \abs{\psi^-_N- f}^2$  & $ \int_{F(S_2^{\mathrm{aff}})} \abs{\psi^+_N- f}^2$ & $ \int_{F(S_2^{\mathrm{aff}})} \abs{\psi^{\mathrm{II},-}_N- f}^2$ & $ \int_{F(S_2^{\mathrm{aff}})} \abs{\psi^{\mathrm{II},+}_N- f}^2$ \\
\hline\hline
$4$ & $97987\cdot 10^{-7}$ & $97336\cdot 10^{-7}$ & $94170\cdot 10^{-7}$ & $89002\cdot 10^{-7}$ \\ \hline
$5$ & $86234\cdot 10^{-7}$ & $86224\cdot 10^{-7}$ & $77865\cdot 10^{-7}$& $77839\cdot 10^{-7}$\\ \hline
$6$ & $21116\cdot 10^{-7}$ & $21447\cdot 10^{-7}$ & $35708\cdot 10^{-7}$& $35636\cdot 10^{-7}$\\ \hline
$7$ & $9841\cdot 10^{-7}$ & $9812\cdot 10^{-7}$ & $14023\cdot 10^{-7}$& $13915\cdot 10^{-7}$\\ \hline
$8$ & $1949\cdot 10^{-7}$ &  $1978\cdot 10^{-7}$ & $2570\cdot 10^{-7}$& $2570\cdot 10^{-7}$\\ \hline
$9$ & $1000\cdot 10^{-7}$ & $1001\cdot 10^{-7}$ & $1309\cdot 10^{-7}$& $1310\cdot 10^{-7}$\\ \hline
$10$ & $503\cdot 10^{-7}$ & $504\cdot 10^{-7}$ &$600\cdot 10^{-7}$& $601\cdot 10^{-7}$ \\ \hline
$11$ & $63\cdot 10^{-7}$ & $63\cdot 10^{-7}$ & $86\cdot 10^{-7}$& $86\cdot 10^{-7}$\\ \hline
$12$ & $3\cdot 10^{-7}$ & $3\cdot 10^{-7}$ & $11\cdot 10^{-7}$& $11\cdot 10^{-7}$ \\ \hline
\end{tabular}
\medskip

\caption{Comparison of errors of interpolations $\psi^-_N$, $\psi^+_N$, $\psi^{\mathrm{II},-}_N$ and $\psi^{\mathrm{II},+}_N$. The function $f$, given by (\ref{f4}), is sampled on grids $L^\pm_{0,\frac{1}{2},N,1}$, $N=4,\dots,12$. } \label{tabcomp}
\end{table}

\item Symmetric cosine interpolating functions $\psi^{\mathrm{I},+}_N$ and $\psi^{\mathrm{II},+}_N$ offer an additional advantage due to the absence of the so called Gibbs phenomenon alongside the borders of $F(S_2^{\mathrm{aff}})$. The Gibbs phenomenon for one-dimensional Fourier expansions, which describes specific behavior ('ringing') of Fourier expansions at the points of discontinuities of $f$ was extensively studied, see e.g. the detailed review~\cite{Hewitt}. The Gibbs phenomenon in dimension two was also investigated in ~\cite{Ustina,Helmberg}. The analog of the Gibbs phenomenon appears in Fourier interpolating functions~\cite{Helmberg2}. Since the (anti)symmetric and periodic extensions of continuous functions on $F(S_2^{\mathrm{aff}})$ may generate discontinuities of the resulting functions on the borders of $F(S_2^{\mathrm{aff}})$, the Gibbs phenomenon is in general to be expected for the interpolating functions
\begin{itemize}
\item $\psi^-_N$: on the borders $y=0$, $x=1$ and $x=y$
\item $\psi^+_N$: on the borders $y=0$, $x=1$
\item $\psi^{\mathrm{I},-}_N$ and $\psi^{\mathrm{II},-}_N$: on the border $x=y$.
\end{itemize}

Since for the continuous function $f' $ on $\overline{F(S_2^{\mathrm{aff}})}$, the periodic extension of $E_1Sf'$ (see (\ref{sext}),(\ref{extE})) is continuous, the Gibbs phenomenon for the interpolations $\psi^{\mathrm{I},-}_N$ and $\psi^{\mathrm{II},-}_N$ of $f'$ does not occur. This fact is illustrated in Figure~\ref{FfCp2}.

\begin{figure}[!ht]
\resizebox{2.4cm}{!}{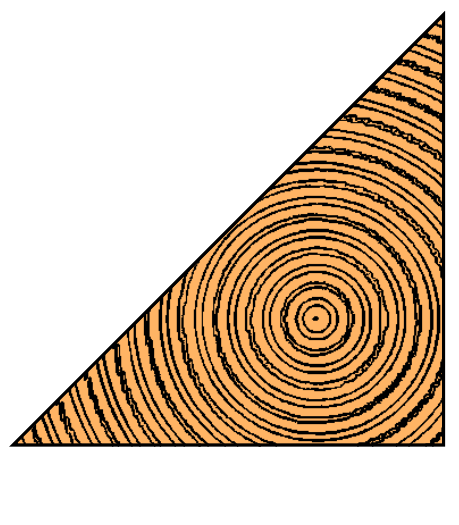}\hspace{1.3cm}
\resizebox{2.4cm}{!}{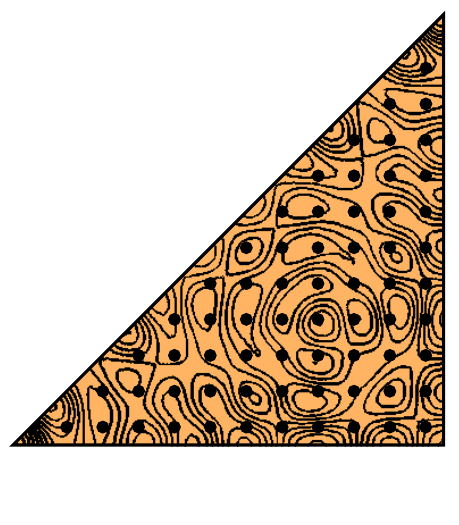}\hspace{1.3cm}
\resizebox{2.4cm}{!}{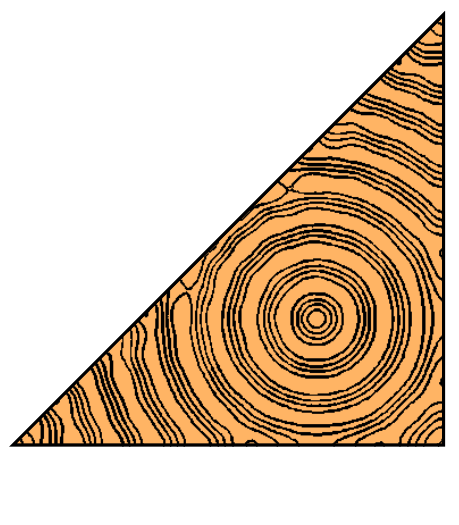}
\\\vspace{1pt}
\resizebox{3.2cm}{!}{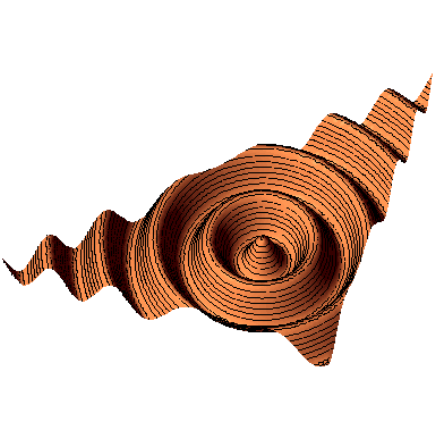}\hspace{14pt}
\resizebox{3.2cm}{!}{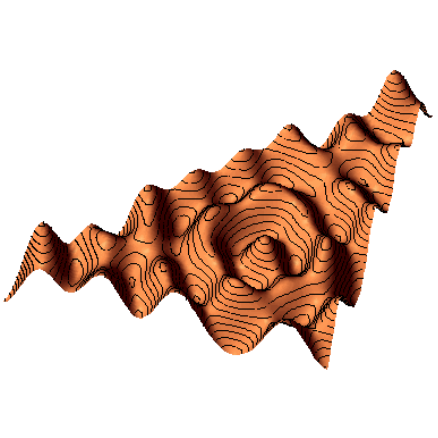}\hspace{14pt}
\resizebox{3.2cm}{!}{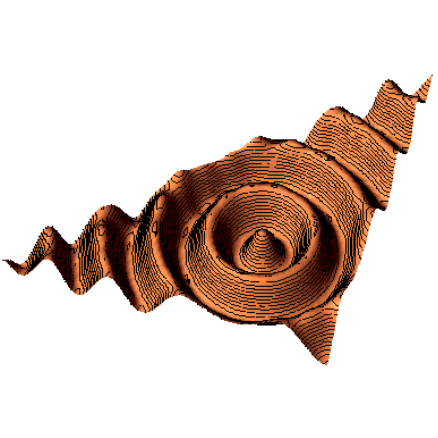}
\caption{The symmetric interpolating cosine functions of type SMDCT-II: $\psi^{\mathrm{II},+}_{12}$, $\psi^{\mathrm{II},+}_{20}$ of the function $f'(x,y)=\frac{1}{15}\cos{40\sqrt{(x-0.707)^2+(y-0.293)^2}}$.}\label{FfCp2}
\end{figure}
\end{enumerate}

\section{Concluding remarks}

When functions of only two variables are considered, the structure of our special functions is quire transparent. For this reason, we avoided any reference to the underlying symmetric group $S_3$ of permutations of three elements \cite{NPT}.

Section~4 contains a description of symmetric and antisymmetric $2D$ transforms built on cosine functions. An analogous presentation of the properties of the transforms built on sine functions would be of interest.

None of the six families of special functions considered in the paper correspond to the well known cosine or sine transforms \cite{Strang,mart}, where one dimensional transforms are used simultaneously in two mutually orthogonal directions. Nevertheless it is justifiable to claim that here we have $2D$ generalizations of common trigonometric functions of one dimension. It is curious to see that both sine and cosine each have a symmetric and antisymmetric generalization.

The conventional cosine transform used in $2D$ comes in four types \cite{Strang}, traditionally labeled I, II, III, and IV. They differ by the shift of the data sampling points with respect to the original lattice points.  The shifting is by a half distance separating the lattice points in either of the two space directions or in both. 

A comparison of theoretical properties of the `cosine transforms of $S_3$' in this paper and those of the standard $S_2\times S_2$ would be of interest and has yet to be made. The arguments in favor of the $S_3$ version of the transforms quote the ease of generalization to any dimension \cite{KPexp} and the possibility to work with data on lattices of other symmetries \cite{NPT}. In general the greater the symmetry group underlying the formalism, the more economies one may expect in some applications. On the down side, there is the need to work with functions situated in domains that are of triangular shape even in square lattices.

Practical computational aspects of the formalism presented in this paper need further investigation. The identification of problems in which it can be most advantageous, a comparison of computation speeds, and modification to the fast Fourier transform, etc. should be addressed.

Visual inspection of the interpolations of the Gaussian model functions in this paper lead to the qualitative conclusion that the interpolation error is rather small once the minimal distances of the lattice grid become smaller than the dispersion of the Gaussian model function.

The orbit functions of the paper have other useful properties that were not exposed here, two of which are:  (i) Their products are decomposable into their finite sums. (ii) They can be naturally rewritten in terms of variables referring to non-orthogonal bases related to the simple roots of the simple Lie group $SU(3)$, see \cite{NPT}.

\section*{Acknowledgements}
Work was supported in part by the Natural Sciences and Engineering Research Council of Canada and by the MIND Research Institute, California. J.~H.~is grateful for the postdoctoral fellowship and the hospitality extended to him at the Centre de recherches math\'ematiques, Universit\'e de Montr\'eal.


\begin{thebibliography}{99}
\bibitem{Stoer}
R. Bulirsch, J. Stoer: {\it Introduction to numerical analysis} 2ed., Springer, (1993)

\bibitem{Davis}
P. Davis : {\it Interpolation and Approximation}, Dover Publications Inc., (1975)
\bibitem{Helmberg2}
G. Helmberg, {\it The Gibbs phenomenon for Fourier interpolation}. J. Approx. Theory {\bf 78} (1994), pp. 41--63.

\bibitem{Helmberg}
G. Helmberg, {\it Localization of a corner-point Gibbs phenomenon for Fourier series in two dimensions}, J. Fourier Anal. Appl. {\bf 8} (2002), 29--41

\bibitem{Hewitt}
E. Hewitt, R. Hewitt, {\it The Gibbs-Wilbraham phenomenon: An episode in Fourier analysis}, Archive for History of Exact Sciences, {\bf 21}, Issue 2, (1979) 129 -- 160.

\bibitem{KPexp}
A. Klimyk, J. Patera: {\it (Anti)symmetric multivariate exponential functions and corresponding Fourier transforms}, J. Phys. A: Math. Theor. {\bf 40} (2007), 10473--10489

\bibitem{KPtrig}
A. Klimyk, J. Patera: {\it (Anti)symmetric multivariate trigonometric functions and corresponding Fourier transforms}, J. Math. Phys. {\bf 48} (2007), 093504

\bibitem{KPalte}
A. Klimyk, J. Patera, {\it Alternating group and multivariate
exponential functions,\/}  in {\sl Groups and Symmetries; from the
Neolithic Scots to John McKay,}  AMS-CRM Proceedings and Lectures Notes
Series, eds. J. Harnad and P. Winternitz (to appear (2008));
arXiv:0907.0601v1

\bibitem{KPalt}
A. Klimyk, J. Patera: {\it Alternating multivariate trigonometric functions and corresponding Fourier transforms}, J. Phys. A: Math. Theor. {\bf 41} (2008), 145205

\bibitem{KP1}
A. Klimyk, J. Patera, {\it Orbit functions,\/}  SIGMA (Symmetry,
Integrability and Geometry: Methods and Applications) {\bf 2} (2006), 006, 60
pages, math-ph/0601037

\bibitem{KP2}
A. Klimyk, J. Patera, {\it Antisymmetric orbit functions,\/}  SIGMA (Symmetry, Integrability and Geometry: Methods and Applications) {\bf 3} (2007), paper 023, 83 pages;  math-ph/0702040v1

\bibitem{KP3}
A. Klimyk, J. Patera, {\it $E$-orbit functions,\/} SIGMA (Symmetry, Integrability and Geometry: Methods and Applications) {\bf 4} (2008), 002, 57 pages; arXiv:0801.0822

\bibitem{mart}
S.~Martuchi, {\it Symmetric convolution and the discrete sine and cosine
transforms,\/} IEEE Trans. Signal Processing,  {\bf 42} (1994) 1038-1051.

\bibitem{Marco}
A. Marco, J. Martinez : {\it Parallel computation of determinants of matrices with polynomial entries}, J. Symb. Comp. {\bf 37} (2004), 749--760

\bibitem{minc}
H. Minc, {\sl Permanents,\/} Addison-Wesley, New York 1978

\bibitem{NPT}
M. Nesterenko, J. Patera, A. Tereszkewicz, {\it Orbit funcions of $SU(n)$ and Chebyshev polybomials,\/} arXiv:0905.2925.

\bibitem{Ustina}
F. Ustina, {\it Gibbs phenomenon for functions of two variables}, Trans. Amer. Math. Soc. {\bf 129} (1967), 124-–129.

\bibitem{Strang}
G.~Strang, {\it The discrete cosine transform,\/}
SIAM Review,  {\bf 41} (1999) 135-147.



\end{thebibliography}
\end{document}